\documentclass[letterpaper,11pt,onecolumn,oneside]{IEEEtran}
\usepackage{ams,times,graphicx,epsfig,epic,eepic,latexsym}

\def\BibTeX{{\rmfamily B\kern-.05em{\scshape i\kern-.025em b}\kern-.08em
\TeX}}

\usepackage{keywords}
\usepackage{cite}
\usepackage{algorithm}
\usepackage{algorithmic}
\usepackage{amsmath,amssymb}
\newtheorem{theorem}{Theorem}
\newtheorem{definition}{Definition}
\newtheorem{lemma}{Lemma}

\newtheorem{corollary}{Corollary}

\newtheorem{condition}{Condition}

\thinlines

\title{Network error correction with unequal link capacities}

\author{\thanks{This work was supported in part by subcontract \#069153 issued by BAE
Systems National Security Solutions, Inc. and supported by the Defense
Advanced Research Projects Agency (DARPA) and the Space and Naval
Warfare System Center (SPAWARSYSCEN), San Diego under Contract No.
N66001-08-C-2013, NSF grant CNS 0905615 and Caltech's
Lee Center for Advanced Networking. Part of this work was performed while A. S. Avestimehr was with the Center for Mathematics of Information, Caltech. The work of A. S. Avestimehr was
partly supported by NSF CAREER award 0953117.}Sukwon Kim, Tracey Ho,~\IEEEmembership{Member,~IEEE}\thanks{Sukwon Kim, Tracey Ho and Michelle Effros are with the Department of Electrical Engineering, California Institute of Technology,
          Pasadena, CA 91125, USA, e-mail: \{sukwon,tho,effros\}@caltech.edu}, Michelle Effros,~\IEEEmembership{Fellow,~IEEE}, Amir Salman Avestimehr,~\IEEEmembership{Member,~IEEE}\thanks{A. S. Avestimehr is with the School of Electrical and Computer Engineering,
Cornell University,
Ithaca, NY, 14853, USA, e-mail: avestimehr@ece.cornell.edu}}

\begin{document}

\maketitle

\begin{abstract}
This paper studies the capacity of single-source single-sink noiseless
networks under adversarial or arbitrary errors on no more than $z$ edges. Unlike
prior papers, which assume equal capacities on all links, 
arbitrary link capacities are considered. Results include new upper bounds, network error correction coding
strategies, and examples of network families where our bounds are
tight. An example is provided of a network where the capacity is 50\% greater than the best rate
that can be achieved with linear coding. While coding at the source
and sink suffices in networks with equal link capacities, in
networks with unequal link capacities, it is shown that
intermediate nodes may have to do coding, nonlinear error detection, or error
correction in order to achieve the network error correction capacity.
\end{abstract}
\section{Introduction}
Network coding allows intermediate nodes in a network to mix the
information content from different packets. This mixing can increase
throughput and reliability in networks of error-free or
stochastically failing
channels~\cite{ahlswede2000network,li2003linear}. Unfortunately, it
can also potentially increase the impact of malicious links or nodes
that wish to corrupt data transmissions. A single corrupted packet,
mixed with other packets in the network, can potentially corrupt all
of the information reaching a particular destination. To combat this
problem, network error correction was first studied by Yeung and
Cai~\cite{yeung2006network,cai2006network} who investigated
correction of errors in multicast network
coding~\cite{ahlswede2000network,li2003linear,koetter2003algebraic}
on networks with unit-capacity links. In that work, the authors
showed that for any network of unit-capacity links, the Singleton
bound is tight and linear network error-correcting codes suffice to
achieve the capacity, which equals $C-2z$ where $C$ is the min-cut
of the network and $z$ is a bound on the number of corrupted
links~\cite[Theorem 4]{cai2006network}. The problem of network
coding under Byzantine attack was also investigated
in~\cite{ho2004byzantine}, which gave an approach for detecting
adversarial errors under random network coding. Construction of
codes that can correct errors up to the full error-correction
capability specified by the Singleton bound was presented
in~\cite{zhang2008linear}. A variety of alternative models of
adversarial attack and strategies for detecting and correcting such
errors appear in the literature. Examples
include~\cite{gkantsidis2006cooperative,krohn2004fly,jaggi16resilient,liang-non,liang-watch,koetter2008coding,silva2008rank,langberg2009binary}.

Specifically, the network error correction problem concerns reliable
information transmission in a network with an adversary who
arbitrarily corrupts the packets sent on some set of $z$ links. The
location of the adversarial links is fixed for all time but unknown
to the network user. We define a $z$-error correcting code
for a single-source and single-sink network to be a code that can
recover the source message at the sink node if there are at most $z$
adversarial links in the network.
The $z$-error correcting network capacity, henceforth simply called
the capacity, is the supremum over all rates achievable by $z$-error
correcting codes.

In this work, we consider network error correction when
links in the network may have unequal link capacities. (A related
model, where adversaries control a fixed number of nodes rather than
a fixed number of edges was studied in~\cite{kosut2009nonlinear}, independently and concurrently with our initial conference paper~\cite{kim2009network}.) The unequal link capacity
problem is substantially different from the equal link capacity problem studied by Yeung
and Cai in~\cite{yeung2006network,cai2006network} since the rate
controlled by the adversary varies with his edge choice. In the error-free case, any link $l$ in the network with capacity
$r$ can be represented by $r$ edges of capacity one without loss of generality. However, in the case with errors there is a loss of generality in using a similar representation and assuming that errors have uniform rate, since this does not capture potential trade-offs that the adversary faces in choosing whether to attack strategically positioned or larger capacity links. The error-correction capacity in the equal link capacity case has a simple cut-set characterization since the adversary always finds it optimal to attack links on a minimum cut; as a result, coding only at the source
and forwarding at intermediate nodes suffices to achieve the
capacity for any single-source and single-sink network. In contrast,
for networks with unequal link capacities, we show that network
error correction coding operations at intermediate nodes are needed even in the single-source single-sink case.

The cut-set approach is a simple yet powerful tool for bounding the
capacity of a large network.  This approach partitions the nodes
into two subsets, say $S$ and $S^c$, and then bounds the rate that
can be transmitted from nodes in $S$ to nodes in $S^c$. (See, for
example,~\cite[Section~15.10]{cover1991elements}.) The maximum
information transmission across the ``cut'' occurs when the nodes
within $S$ can collaborate perfectly among themselves and the nodes
within $S^c$ can collaborate perfectly among themselves. In this
case, $S$ and $S^c$ each act as ``super-nodes'' in a simple
point-to-point network. All that is needed for collaboration is
sufficient information exchange among the nodes on each side of the
cut. Thus, the ``cut-set bound'' equals that rate that would be
achieved in transmitting information from $S$ to $S^c$ if we added
reliable, infinite-capacity links between each pair of nodes in $S$
and reliable, infinite-capacity links between each pair of nodes in
$S^c$, as shown in Figure~\ref{fig:cs}. Given a network of
capacitated error-free links with a  source node $s$ and a sink node
$t$, minimizing over all choices of $S$ that contain $s$ but exclude
$t$ gives a tight bound on the unicast capacity from $s$ to
$t$~\cite{ahuja-network}.

In contrast, this traditional
cut-set bounding approach is not tight in general when it comes to the error-correction capacity of networks with unequal link capacities,
even in the case of unicast demands.
In this case, two new issues arise.
We next describe each of these issues in turn.

\begin{figure}
\begin{center}
\begin{picture}(110,100)(-10,-10)
\thicklines
\put(20,60){\circle*{3}}
\put(50,90){\circle*{3}}
\put(80,60){\circle*{3}}
\put(20,30){\circle*{3}}
\put(50,0){\circle*{3}}
\put(80,30){\circle*{3}}
\thinlines
\put(50,90){\vector(-1,-1){30}}
\put(50,90){\vector(1,-1){30}}
\put(20,60){\vector(0,-1){30}}
\put(20,30){\vector(2,1){60}}
\put(20,30){\vector(1,-1){30}}
\put(80,60){\vector(-1,-2){30}}
\put(80,60){\vector(-0,-1){30}}
\put(80,30){\vector(-1,-1){30}}
\thinlines
\put(-10,45){\multiput(0,0)(5,0){22}{\line(1,0){3}}}
\put(-10,50){\makebox(0,0)[lb]{\small $S$}}
\put(-10,40){\makebox(0,0)[lt]{\small $S^c$}}
\put(50,93){\makebox(0,0)[cb]{\small $s$}}
\put(50,-3){\makebox(0,0)[ct]{\small $t$}}
\put(50,-15){\makebox(0,0)[ct]{(a)}}
\end{picture}
\hspace{.05in}
\begin{picture}(110,100)(-10,-10)
\thicklines
\put(20,60){\circle*{3}}
\put(50,90){\circle*{3}}
\put(80,60){\circle*{3}}
\put(20,30){\circle*{3}}
\put(50,0){\circle*{3}}
\put(80,30){\circle*{3}}
\thicklines
\put(50,90){\vector(-1,-1){30}}
\put(50,89.5){\vector(-1,-1){30}}
\put(50,90){\vector(1,1){0}}
\put(50,89.5){\vector(1,1){0}}
\put(50,90){\vector(1,-1){30}}
\put(50,89.5){\vector(1,-1){30}}
\put(50,90){\vector(-1,1){0}}
\put(50,89.5){\vector(-1,1){0}}
\put(20,60){\vector(1,0){60}}
\put(20,59.5){\vector(1,0){60}}
\put(20,60){\vector(-1,0){0}}
\put(20,59.5){\vector(-1,0){0}}
\put(20,30){\vector(1,-1){30}}
\put(20,29.5){\vector(1,-1){30}}
\put(20,30){\vector(-1,1){0}}
\put(20,29.5){\vector(-1,1){0}}
\put(80,30){\vector(-1,-1){30}}
\put(80,29.5){\vector(-1,-1){30}}
\put(80,30){\vector(1,1){0}}
\put(80,29.5){\vector(1,1){0}}
\put(20,30){\vector(1,0){60}}
\put(20,29.5){\vector(1,0){60}}
\put(20,30){\vector(-1,0){0}}
\put(20,29.5){\vector(-1,0){0}}
\thinlines
\put(20,60){\vector(0,-1){30}}
\put(20,30){\vector(2,1){60}}
\put(80,60){\vector(-1,-2){30}}
\put(80,60){\vector(-0,-1){30}}
\thinlines
\put(-10,45){\multiput(0,0)(5,0){22}{\line(1,0){3}}}
\put(-10,50){\makebox(0,0)[lb]{\small $S$}}
\put(-10,40){\makebox(0,0)[lt]{\small $S^c$}}
\put(50,93){\makebox(0,0)[cb]{\small $s$}}
\put(50,-3){\makebox(0,0)[ct]{\small $t$}}
\put(50,-15){\makebox(0,0)[ct]{(b)}}
\end{picture}
\hspace{.05in}
\begin{picture}(110,100)(-10,-10)
\thicklines
\put(20,60){\circle*{3}}
\put(50,90){\circle*{3}}
\put(80,60){\circle*{3}}
\put(20,30){\circle*{3}}
\put(50,0){\circle*{3}}
\put(80,30){\circle*{3}}
\thicklines
\put(50,90){\vector(-1,-1){30}}
\put(50,89.5){\vector(-1,-1){30}}
\put(20,60){\vector(1,0){60}}
\put(20,59.5){\vector(1,0){60}}
\put(20,30){\vector(1,0){60}}
\put(20,29.5){\vector(1,0){60}}
\put(80,30){\vector(-1,-1){30}}
\put(80,29.5){\vector(-1,-1){30}}
\thinlines
\put(50,90){\vector(1,-1){30}}
\put(20,60){\vector(0,-1){30}}
\put(20,30){\vector(2,1){60}}
\put(80,60){\vector(-1,-2){30}}
\put(80,60){\vector(-0,-1){30}}
\put(20,30){\vector(1,-1){30}}
\thinlines
\put(-10,45){\multiput(0,0)(5,0){22}{\line(1,0){3}}}
\put(-10,50){\makebox(0,0)[lb]{\small $S$}}
\put(-10,40){\makebox(0,0)[lt]{\small $S^c$}}
\put(50,93){\makebox(0,0)[cb]{\small $s$}}
\put(50,-3){\makebox(0,0)[ct]{\small $t$}}
\put(50,-15){\makebox(0,0)[ct]{(c)}}
\end{picture}
\hspace{.05in}
\begin{picture}(140,100)(5,-10)
\thicklines
\put(20,60){\circle*{3}}
\put(80,60){\circle*{3}}
\put(20,30){\circle*{3}}
\put(80,30){\circle*{3}}
\put(140,30){\circle*{3}}
\thicklines
\put(20,60){\vector(1,0){60}}
\put(20,59.5){\vector(1,0){60}}
\put(20,30){\vector(1,0){60}}
\put(20,29.5){\vector(1,0){60}}
\put(80,30){\vector(1,0){60}}
\put(80,29.5){\vector(1,0){60}}
\thinlines
\put(20,60){\vector(0,-1){30}}
\put(20,30){\vector(2,1){60}}
\put(80,60){\vector(2,-1){60}}
\put(80,60){\vector(0,-1){30}}
\thinlines
\put(5,45){\multiput(0,0)(5,0){28}{\line(1,0){3}}}
\put(5,50){\makebox(0,0)[lb]{\small $S$}}
\put(5,40){\makebox(0,0)[lt]{\small $S^c$}}
\put(140,27){\makebox(0,0)[ct]{\small $t$}}
\put(80,-15){\makebox(0,0)[ct]{(d)}}
\end{picture}
\end{center}
\caption{The traditional cut set bound for a cut $S$
between source $s$ and sink $t$
in the network shown in~(a)
equals the maximal rate that can be transmitted
from $S$ to $S^c$ when the nodes within $S$
are allowed unlimited information exchange
and the nodes within $S^c$ are allowed unlimited information exchange,
as indicated by the thick bidirectional lines in~(b).
In the cut-set bounds employed here,
we create infinite capacity connections
only from each node
to nodes of higher topological order
on the same side of the cut, as shown by the thick unidirectional lines in~(c).
Restricting attention to nodes with input or output edges
that cross the cut, gives a ``zig-zag'' network,
as shown in~(d).}\label{fig:cs}
\end{figure}
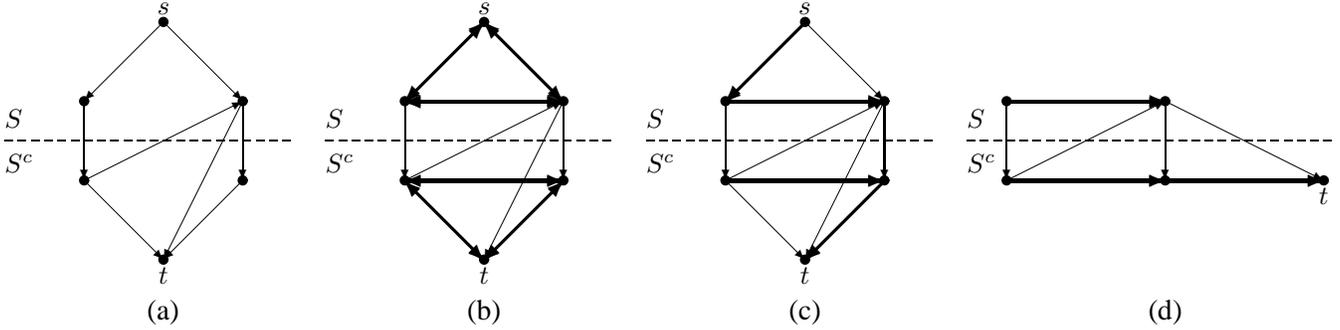

The first issue concerns the role of feedback across $S$ --
i.e. links from $S^c$ to $S$.
While feedback never increases the capacity
across a cut in a network of reliable links,
it can increase the error-correction capacity.
Intuitively, this is because feedback
allows us to inform nodes in $S$
about what was received by nodes in $S^c$,
thereby aiding in the discovery of adversarial links.\footnote{This
process of discovery is complicated by the fact that the feedback links themselves may be corrupted,
but feedback is, nonetheless, clearly useful.}
Treating all nodes in $S$ as one super-node
and all nodes in $S^c$ as another super-node, as in the traditional cut-set bounding approach,
makes all feedback information available to all nodes in $S$
and all feedforward information available to all nodes in $S^c$.  This may give them considerably more insight
into the adversary's location than is available to them
in the original network.

We can obtain tighter bounds by taking into account limitations on which nodes in $S^c$ can influence the values on each feedback edge and which nodes in $S$ have access to the feedback information. This is important in the unequal link capacity case, because it captures trade-offs faced by the adversary in choosing whether to attack links based on their capacity or whether they are upstream of feedback links that may give clues about the adversary's actions. Specifically, given an acyclic network $\mathcal{G}$, we construct an acyclic network $\mathcal{G}'$
by adding a reliable infinite capacity connection
from a node $v\in S$ to a node $w\in S$
only if $\mathcal{G}$ contains a directed path from node $v$ to node $w$ via nodes in $S$, and adding  a reliable infinite capacity connection
from a node $v\in S^c$ to a node $w\in S^c$
only if $\mathcal{G}$ contains a directed path from node $v$ to node $w$ via nodes in $S^c$.
Figure~\ref{fig:cs}(c) shows an example.
Limiting the added connections in this way
creates what we call a ``zig-zag'' network,
as shown in Figure~\ref{fig:cs}.
We draw only those nodes in $S$ and $S^c$
with incoming or outgoing edges that cross between $S$ and $S^c$, and draw the nodes on each side of the cut in topologically increasing order.
The ``forward'' edges across the cut point downwards in the diagram,
while ``feedback'' edges point upwards.  By studying the capacity of these zig-zag networks, we develop upper bounds  on error-correction capacity that apply to general acyclic networks.  We also illustrate the usefulness of these bounds by giving examples where they improve upon previously known bounds  and showing that they are tight for families of networks that are special cases
of zig-zag networks.

However, the second issue
with the cut-set approach
to bounding network capacities
is the notion of a cut itself.
Reference~\cite{kosut-polytope} shows, for the more general case where only a subset of links are potentially adversarial,
the existence of networks for which
no partition $(S,S^c)$ yields a tight bound
on the error-correction capacity.
This is proven by example using a network
whose minimal cut
(which has no feedback links)
yields a capacity bound that is proven to be unachievable.
As a result, knowledge of the the capacity of the network's minimal cut
is insufficient to determine the capacity of all possible networks,
and we cannot hope to derive cut-set bounds
that are tight in general.
Nonetheless, given the complexity of taking into account the full network topology,
we proceed to study the cut-set approach, deriving general bounds and demonstrating
that those bounds are tight in some cases.

Specifically, in Section III we begin
with the cut-set upper bound given by the capacity of the two-node network shown in Fig.~\ref{fig61},
which is the only cyclic network we consider in this paper. In this
network, the source node can transmit packets to the sink node along
the forward links and the sink node can send information back to the
source node along the feedback links. As mentioned above, this cut-set bound can be quite loose since it assumes that all feedback is available
to the source node and all information crossing the cut in the
forward direction is available to the sink. We therefore develop a new cut-set upper bound for general acyclic networks by taking into account more details of the topological relationships among links on the cut, as in the zig-zag network construction shown in Figure~\ref{fig:cs}.

In Section IV, we consider a variety
of linear and nonlinear coding strategies useful for achieving the
capacity of various example networks. We 
prove the insufficiency of linear network codes to achieve the
capacity by providing an example of a network where the capacity is 50\% greater than the
linear coding capacity and is achieved using a nonlinear error
detection strategy. A similar example for the problem with Byzantine
attack on nodes rather than edges appears
in~\cite{kosut2009nonlinear}. We also give examples of single-source
and single-sink networks for which intermediate nodes must perform
coding, nonlinear error detection or error correction in order to
achieve the network capacity.  We describe a simple greedy algorithm
for error correction at intermediate nodes.
We then introduce a new
coding strategy called ``guess-and-forward." In this strategy, an
intermediate node which receives some redundant information from
multiple paths guesses which of its upstream links controlled by the
adversary. The intermediate node forwards its guess to the sink
which tests the hypothesis of the guessing node.  In Section V, we
show that guess-and-forward achieves network capacity on the
two-node network with feedback links of  Fig.~\ref{fig61}, as well as the
family of four-node acyclic
networks in Fig.~\ref{fig62} when the capacity of each feedback link is not too small (i.e.~above a value given by a linear optimization).\footnote{After the submission of this paper, we obtained a new result that improves upon the bound in Section III for the special case of small-capacity feedback links. We mention the idea briefly as a footnote in Section III and will present it formally in an upcoming paper.}
Finally, we apply guess-and-forward strategy to zig-zag
networks, deriving achievable rates and presenting
conditions under which our upper bound is tight. We conclude in Section
VI with a discussion of future work.

Portions of this work have
appeared in our earlier work~\cite{kim2009network,kim2010nec}, which introduced
the network error correction problem with unequal link
capacities and presented a subset of the results.

\begin{figure}
\begin{center}
\begin{picture}(200,200)(-40,-20)\centering
\put(75,150){\circle{7}}\put(75,15){\circle{5}}\qbezier(75,150)(-20,82)(75,15)
\qbezier(75,150)(20,82)(75,15)\qbezier(75,150)(130,82)(75,15)\qbezier(75,150)(170,82)(75,15)
\put(30,82){\circle{1}}\put(36,82){\circle{1}}\put(42,82){\circle{1}}
\put(108,82){\circle{1}}\put(114,82){\circle{1}}\put(120,82){\circle{1}}
\put(36,72){$n$}\put(110,72){$m$}
\put(75,160){$s$}\put(75,3){$t$}
\put(28,80){\vector(0,-1){0}}\put(48,80){\vector(0,-1){0}}
\put(103,80){\vector(0,1){0}}\put(123,80){\vector(0,1){0}}
\end{picture}
\end{center}
\caption{Two-node network composed of $n$ forward links and $m$
feedback links.}\label{fig61}
\end{figure}
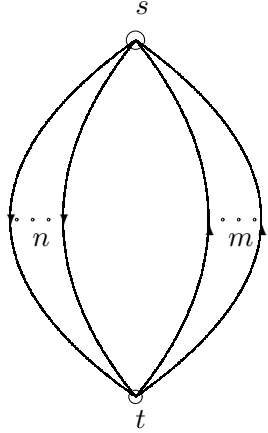

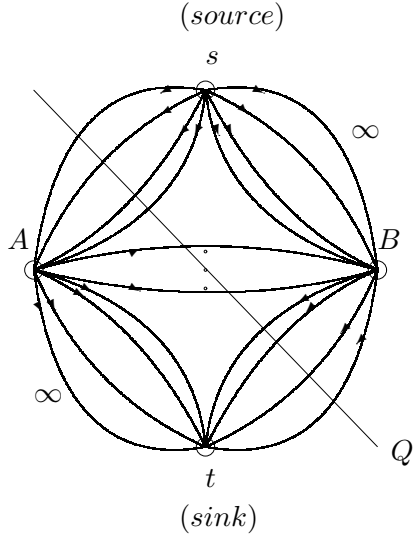
\begin{figure}
\begin{center}
\begin{picture}(200,200)(-40,-20)\centering
\put(75,150){\circle{7}}\put(75,15){\circle{7}}\put(10,82){\circle{7}}\put(140,82){\circle{7}}
\qbezier(75,150)(50,100)(10,82)\qbezier(75,150)(70,100)(10,82)
\qbezier(75,150)(20,160)(10,82)\qbezier(75,150)(30,130)(10,82)
\qbezier(75,150)(100,100)(140,82)\qbezier(75,150)(80,100)(140,82)
\qbezier(75,150)(130,160)(140,82)\qbezier(75,150)(120,130)(140,82)
\qbezier(10,82)(50,65)(75,15)\qbezier(10,82)(70,65)(75,15)
\qbezier(10,82)(20,5)(75,15)\qbezier(10,82)(30,35)(75,15)
\qbezier(140,82)(100,65)(75,15)\qbezier(140,82)(80,65)(75,15)
\qbezier(140,82)(130,5)(75,15)\qbezier(140,82)(120,35)(75,15)
\qbezier(10,82)(75,100)(140,82)\qbezier(10,82)(80,65)(140,82)
\put(75,89){\circle{1}}\put(75,82){\circle{1}}\put(75,75){\circle{1}}
\put(75,160){$s$}\put(65,175){$(source)$}\put(75,0){$t$}\put(65,-15){$(sink)$}
\put(0,90){$A$}\put(140,90){$B$}\put(10,32){$\infty$}\put(130,132){$\infty$}
\put(58,149){\vector(-2,-1){0}}\put(58,141){\vector(-2,-1){0}}
\put(65,133){\vector(-1,-1){0}}\put(71,133){\vector(-1,-2){0}}
\put(97,149){\vector(2,-1){0}}\put(85,133){\vector(1,-2){0}}
\put(91,141){\vector(1,-1){0}}\put(80,133){\vector(1,-3){0}}
\put(32,74){\vector(2,-1){0}}\put(13,66){\vector(1,-2){0}}
\put(29,72){\vector(2,-1){0}}\put(18,67){\vector(1,-2){0}}
\put(50,90){\vector(2,1){0}}\put(50,74){\vector(2,-1){0}}
\put(110,70){\vector(-2,-1){0}}\put(113,66){\vector(-1,-1){0}}
\put(126,58){\vector(-1,-1){0}}\put(132,54){\vector(-1,-1){0}}
\put(10,150){\line(26,-27){130}}\put(145,10){$Q$}
\end{picture}
\end{center}
\caption{Four node acyclic networks: Given the cut
$Q=cut(\{s,B\},\{A,t\})$, unbounded reliable communication is
allowed from source $s$ to its neighbor $B$ on one side of the cut
and from node $A$ to sink $t$ on the other side of the cut,
respectively. There are feedback links from $A$ to
$B$.}\label{fig62}
\end{figure}

\section{Preliminaries}

Consider a directed acyclic communication
network $\mathcal{G} = (\mathcal{V},\mathcal{E})$ with unequal link
capacities. Let $r(l)$ denote the
capacity of edge $l\in\mathcal{E}$. A source node $s\in\mathcal{V}$
transmits information to a  sink node $t\in\mathcal{V}$ over the network $\mathcal{G}$.
Transmissions occur on the links according to their topological order, i.e.~a link $l$ transmits after all its incident incoming links, and we regard a link error as being applied upon transmission. A link (or node) is said to be {\it upstream} of another link (or node) iff there is a directed path starting from the former and ending with the latter.
A link (or node) is said to be {\it downstream} of another link (or node) iff there is a directed path starting from the latter and ending with the former.

In this paper, we consider the problem of correcting arbitrary adversarial errors on up to $z$ links. The location of error links is fixed for all time but unknown
to the network user.

\begin{definition}
A network code is $z$-error link-correcting
if the source message can be recovered by the sink node provided
that the adversary controls at most $z$ links. Thus a $z$-error
link-correcting network code can correct any $\tau$ adversarial
links for $\tau\leq z$.
\end{definition}

Let $(S,S^c)$ be a partition of $\mathcal{V}$, and define the cut for
the partition $(S,S^c)$ by
\[cut(S,S^c)=\{(a,b)\in\mathcal{E}:a\in S, b\in S^c\}.\]
The cut $cut(S,S^c)$ separates nodes $a$ and $b$
if $a\in S$ and $b\in S^c$. We use $CS(a,b)$ to denote the set of cuts
between $a$ and $b$.
Given a cut $Q= cut(S,S^c)$, we call any link in
$Q$ a forward link, and we call any link from
$S^c$ to $S$ a feedback link.

For the achievable strategies in Sections IV and V, we assume that coding occurs in the finite field $\mathbb{F}_q$ for some prime power $q$. An error on
any link $l\in\mathcal{E}$ is specified by a vector $e_l$ containing $r(l)$
symbols in $\mathbb{F}_q$. The output $y_l$ of link $l$ equals
the  sum in $\mathbb{F}_q$ of the input $x_l$ to link $l$ and the error
$e_l$ applied to link $l$, i.e., $y_l=x_l+e_l$. We say that
an error occurs on the link $l$ if $e_l\neq 0$.

As in~\cite{yeung2006network,cai2006network}, we can consider a
linear network code $V$
that assigns 
a
set of $r(l)$ vectors $\{v\mathbb{F}(l)_1,
v\mathbb{F}(l)_2,..,v\mathbb{F}(l)_{r(l)}\}$, called global coding vectors, to each link $l\in \mathcal{E}$
in the network.
Let
\[\tilde{\phi}_l(w) = \{\langle w,v\mathbb{F}(l)_i\rangle : 1\leq i\leq r(l)\}\]
denote the error-free output of link $l$ when the network input is
$w$ where $\langle a,b \rangle$ denotes the inner product of row
vectors $a$ and $b$. We use $e=(e_l:l\in\mathcal{E})$ to denote the vector of errors on the entire
network. The output of a link $l$ is a
function of both the network input $w$ and the error vector $e$, which we denote  by $\psi_l(w,e)$. For each node $v\in\mathcal{V}$, we use
$\Gamma_+(v)=\{(c,v):(c,v)\in\mathcal{E}\}$ and
$\Gamma_-(v)=\{(v,c):(v,c)\in\mathcal{E}\}$ to denote the sets of
incoming and outgoing edges respectively for node $v$.
With this notation, a sink node
$t$ cannot distinguish between the case where $w$ is the network
input and error $e$ occurs and the case where $w'$ is the network
input and error $e'$ occurs if and only if
\begin{equation}\label{eq1}
(\psi_l(w,e):l\in \Gamma_+(t))=(\psi_l(w',e'):l\in \Gamma_+(t)).
\end{equation}
Let $N(e)=|\{l\in\mathcal{E}:e_l\neq 0\}|$ denote the number of
links in which an error occurs. We say that any pair of input
vectors $w$ and $w'$ are $z$-links separable at sink node $t$
if~(\ref{eq1}) does not hold for any pair of error vectors $e$ and
$e'$ such that $N(e)\leq z$ and $N(e')\leq z$. Lemma 1
of~\cite{cai2006network} establishes the linear properties of
$\psi_l(w,e)$ for networks with unit link capacities. This result
extends directly to networks with arbitrary link capacities.

\begin{lemma}\label{lemmapsi}
For all $l\in \mathcal{E}$, all network inputs $w$ and $w'$, error
vectors $e$ and $e'$, and $\mu\in \mathbb{F}_q$,
\[\psi_l(w+w',e+e') =\psi_l(w,e)+\psi_l(w',e')\] and
\[\psi_l(\mu w) = \mu\psi_l(w).\]
\end{lemma}

From Lemma~\ref{lemmapsi},
\[\psi_l(w,e)=\psi_l(w,0)+\psi_l(0,e) =
\tilde{\phi}_l(w)+\theta_l(e),\] where $\theta_l(e)=\psi_l(0,e)$ for
any link $l$. Thus $\psi_l(w,e)$ can be written as the sum of a
linear function of $w$ and a linear function of $e$.

\section{Upper bounds}\label{section:upper}

In this section, we consider upper bounds on network error
correction capacity. Let $X$ denote the source alphabet and $q$ the
size of the (arbitrary) link alphabet. The corresponding network
transmission rate is given by $\frac{\log|X|}{\log q}.$

We first derive the cut-set upper bound obtained from coalescing all
nodes on each side of the cut into a super-node, resulting in a
two-node network as shown in Fig.~\ref{fig61}.

\begin{lemma}\label{twonodeup}
Consider the two-node network shown in Fig.~\ref{fig61} with
arbitrary link capacities. Let $D_p$ denote the sum of the $p$
smallest forward link capacities. The network error correction capacity of this network is upper bounded by
\[ \left\{\begin{array}{ll}
            0 & \mbox{if $n\leq 2z$}\\
           {\min\{D_{n-z},D_{n-2(z-m)^+}\}} &\mbox{if $n>2z$}\\
            \end{array}
             \right.\]
\end{lemma}

\begin{proof}
Case 1) $n\leq 2z$.

Suppose that $C>0$ and we show a contradiction. Since $C>0$, there
are two codewords $x$ and $y$ in $X$ that can be sent reliably. When $x$ is
sent along the forward links and the leftmost $z$ links are
adversarial, the adversary changes $x$ to $x'$ so that the outputs
of the $\lfloor n/2\rfloor$ leftmost links of $x'$ are the same as
that of $y$. Similarly, when $y$ is sent along the forward links and
the rightmost $z$ links are adversarial, the adversary changes $y$
to $y'$ so that the outputs of the rightmost $\lceil n/2\rceil$
links of $y'$ are the same as that of $x$. Then the two codewords
cannot be distinguished and this contradicts $C>0$.

Case 2) $n\geq 2z$.

When the sink knows $z$ adversarial links are the $z$ largest
capacities forward links, the maximum achievable capacity is
$D_{n-z}$. When $m\leq z$ and all $m$ feedback links are
adversarial, there are $z-m$ adversarial forward links whose
locations are unknown. In this scenario, we show that the best
achievable rate is $D_{n-2(z-m)}$, which is the sum of $n-2(z-m)$
smallest forward link capacities. We assume that $|X|>
q^{D_{n-2(z-m)}}$, and show that this leads to a contradiction.
$F=\{l_1,..,l_n\}$ denotes the set of forward links such that the
links indexed in increasing capacity order, i.e., $r(l_1)\leq
\ldots\leq r(l_n)$. Since $|X|>q^{D_{n-2(z-m)}}$ and $D_{n-2(z-m)}$
is sum of the $n-2(z-m)$ smallest forward link capacities, there
exist two distinct codewords $x,x'\in X$ such that
$\tilde{\phi}_{l_{i}}(x) =\tilde{\phi}_{l_{i}}(x')$ $\forall i
=1,..,n-2z$. So we can write
\[O(x) = \{y_1,..,y_{n-2z},p_{1},..,p_z,w_1,..,w_z\},\]
\[O(x') = \{y_1,..,y_{n-2z},p_{1}',..,p_z',w_1',..,w_z'\},\]
where $O(x)$ denotes the error-free vector of symbols on $Q$ when
codeword $x$ is transmitted.

We can construct $z$-error links that changes $O(x)$ to the value
$\{y_1,..,y_{n-2z},p_{1}',..,p_z',w_1,..,w_z\}$ as follows. We apply
an error of value ($p_i'-p_i)\mod q$ on links $l_{n-2z+i}$ for
$1\leq i\leq z$. Since this does not change the output value of
other $n-z$ links, we obtain
$\{y_1,..,y_{n-2z},p_{1}',..,p_z',w_1,..,w_z\}$. For $x'$, we can
follow a similar procedure to construct $z$ error links that change
the value of $O(x')$ to
$\{y_1,..,y_{n-2z},p_{1}',..,p_z',w_1,..,w_z\}$. Thus, sink node $u$
cannot reliably distinguish between the source symbol $x$ and $x'$,
which gives a contradiction.

Therefore, the upper bound on achievable capacity is
$\min\{D_{n-z},D_{n-2(z-m)^+}\}$.
\end{proof}

In Section V we show that this bound is the actual capacity of the
two-node network. Thus, the super-node construction gives the
following cut-set upper bound for general acyclic networks.

\begin{lemma}\label{twonodeupgen}
Given any cut $Q\in CS(s,t)$ with $k$ forward links and $r$ feedback
links, let $D_p$ denote the sum of the $k$ smallest forward link
capacities. The network error correction capacity is upper bounded by
\[\left\{\begin{array}{ll}
            0 & \mbox{if $k\leq 2z$}\\
           {\min\{D_{k-z},D_{k-2(z-r)^+}\}} &\mbox{if $k>2z$}\\
            \end{array}
             \right.\]
\end{lemma}

However, we can show that the above upper bound is not tight using
the following generalized Singleton bound, which was presented in
our conference paper~\cite{kim2009network}. A similar upper bound
for the problem of adversarial attack on nodes rather than edges was
given in independent work~\cite{kosut2009nonlinear}.

\begin{definition}
Any set of links $S$ on a  cut $Q\in CS(s,t)$ is said to satisfy
the {\it downstream condition} on $Q$ if none of the links in
$Q\backslash S$ are downstream of any link in $S$.
\end{definition}

\begin{lemma}\label{singleton}(A generalized Singleton bound) Consider
any $z$-error correcting network code with source alphabet $X$ in an
acyclic network $\mathcal{G}$. Consider any set $S$ consisting of
$2z$ links on a source-sink cut $Q\in CS(s,t)$ that satisfies the
downstream condition on $Q$. Let $M=\sum_{(a,b)\in Q\backslash
S}r(a,b)$ be the total capacity of the links in $Q \backslash S$.
Then
\[\log|X|\leq M\cdot\log q.\]
\end{lemma}

\begin{proof}
The proof is similar to that of the network Singleton bound for the
equal link capacity case in~\cite{yeung2006network}. We assume that
$|X|> q^{M}$, and show that this leads to a contradiction.

Given a cut $Q$, $K(Q)$ denotes the number of links in $Q$. For
brevity, let $Q = \{l_{1},..,l_{{K(Q)}}\}$ where
$S=\{l_{K(Q)-2z+1},...l_{K(Q)}\}$ and links in $S$ are ordered
topologically, i.e., $l_{K(Q)-2z+i}$ is not downstream of
$l_{K(Q)-2z+j}$ for any $i<j$. Since $|X|>q^{M}$ and $M$ is the
capacity of $Q\backslash S$, there exist two distinct codewords
$x,x'\in X$ such that $\tilde{\phi}_{l_{i}}(x)
=\tilde{\phi}_{l_{i}}(x')$ $\forall i =1,..,K(Q)-2z$. So we can
write
\[O(x) = \{y_1,..,y_{K(Q)-2z},p_{1},..,p_z,w_1,..,w_z\},\]
\[O(x') = \{y_1,..,y_{K(Q)-2z},p_{1}',..,p_z',w_1',..,w_z'\},\]
where $O(x)$ denotes the error-free vector of symbols on $Q$ when
codeword $x$ is transmitted.

We will show that it is possible for the adversary to produce
exactly the same outputs on all the channels in $Q$ when errors
occur on at most $z$ links in $Q$.

Assume that the true network input is $x$. The adversary will inject
errors on $z$ links $l_{K(Q)-2z+1},..,l_{K(Q)-z}$ in this order as
follows. First the adversary applies an error on link
$l_{K(Q)-2z+1}$ to change the output from $p_1$ to $p_1'$. The
output of links $(l_{K(Q)-2z+2},..,l_{K(Q)})$ may be affected by
this change, but the outputs of links $(l_{1},..,l_{K(Q)-2z})$ will
not. Let $p_i'(j)$ and $w_i'(j)$ denote the outputs of links
$l_{K(Q)-2z+i}$ and $l_{K(Q)-z+i}$, respectively after the adversary
has injected errors on link $l_{K(Q)-2z+j}$, where $j=1,2,..,z$ with
$p_1'(1)=p_1'$. Then the adversary injects errors on link
$l_{K(Q)-2z+2}$ to change its output from $p_2'(1)$ to $p_2'$. This
process continues until the adversary finishes injecting errors on
$z$ links $l_{K(Q)-2z+1},..,l_{K(Q)-z}$ and the output of this
channel changes from $O(x)$ to
$\{y_1,..,y_{K(Q)-2z},p_{1}',..,p_z',w_1'(z),..,w_z'(z)\}$. Now
suppose the input is $x'$. We can follow a similar procedure by
injecting errors on $z$ links $l_{K(Q)-z+1},..,l_{K(Q)}$. Then the
adversary can produce the outputs
\[\{y_1,..,y_{K(Q)-2z},p_{1}',..,p_z',w_1'(z),..,w_z'(z)\}.\] Thus,
sink node $t$ cannot reliably distinguish between the source symbol
$x$ and $x'$, which gives a contradiction.
\end{proof}

\begin{figure}
\begin{center}
\begin{picture}(200,200)(-40,-20)\centering
\put(75,150){\circle{7}}\put(75,15){\circle{7}}\put(10,82){\circle{7}}\put(140,82){\circle{7}}
\qbezier(75,150)(70,100)(10,82) \qbezier(75,150)(20,160)(10,82)
\qbezier(75,150)(100,100)(140,82)\qbezier(75,150)(80,100)(140,82)
\qbezier(75,150)(130,160)(140,82)\qbezier(75,150)(120,130)(140,82)
\qbezier(10,82)(50,65)(75,15)\qbezier(10,82)(70,65)(75,15)
\qbezier(10,82)(20,5)(75,15)\qbezier(10,82)(30,35)(75,15)
\qbezier(140,82)(100,65)(75,15)\qbezier(140,82)(80,65)(75,15)
\qbezier(140,82)(130,5)(75,15)\qbezier(140,82)(120,35)(75,15)
\qbezier(10,82)(75,100)(140,82)\qbezier(10,82)(80,65)(140,82)
\put(75,160){$s$}\put(65,175){$(source)$}\put(75,0){$t$}\put(65,-15){$(sink)$}
\put(0,90){$A$}\put(140,90){$B$}\put(10,32){$\infty$}\put(130,132){$\infty$}
\put(58,149){\vector(-2,-1){0}}
\put(71,133){\vector(-1,-2){0}}
\put(97,149){\vector(2,-1){0}}\put(85,133){\vector(1,-2){0}}
\put(91,141){\vector(1,-1){0}}\put(80,133){\vector(1,-3){0}}
\put(32,74){\vector(2,-1){0}}\put(13,66){\vector(1,-2){0}}
\put(29,72){\vector(2,-1){0}}\put(18,67){\vector(1,-2){0}}
\put(50,90){\vector(2,1){0}}\put(50,74){\vector(2,-1){0}}
\put(110,70){\vector(-2,-1){0}}\put(113,66){\vector(-1,-1){0}}
\put(126,58){\vector(-1,-1){0}}\put(132,54){\vector(-1,-1){0}}
\put(10,150){\line(26,-27){130}}\put(145,10){$Q$}
\put(30,140){$1$}\put(50,110){$1$}\put(90,60){$10$}
\put(98,52){$10$}\put(110,42){$10$}\put(118,33){$10$}\put(10,82){\line(1,0){130}}
\end{picture}
\end{center}
\caption{Four-node acyclic network: unbounded reliable communication
is allowed from source $s$ to its neighbor $B$ on one side of the
cut and from node $A$ to sink $t$ on the other side of the cut,
respectively. There are 2 forward links of capacity 1 from $s$ to
$A$, 4 forward links of capacity 10 from $B$ to $t$, and 3 feedback
links from $A$ to $B$.}\label{fig152}
\end{figure}
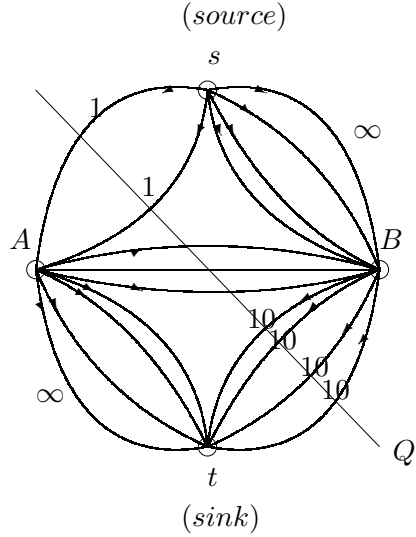

Consider  the example four-node
network shown in Fig.~\ref{fig152}. When $z=2$, the two-node bound
lemma~\ref{twonodeupgen} gives the upper bound 22. The generalized
Singleton bound gives upper bound 2.

However, the generalized Singleton bound is also not tight. Building on ideas from the above bounds, we proceed to derive tighter bounds.

Let
$Q^R$ denote the set of feedback links across cut $Q$. Given a set
of $m\leq z$ feedback links $W\subset Q^R$ and a set of $k\leq z-m$
forward links $F\subset Q$, we use $N_{z,m,k}^{F,W}(Q)$ to denote
the upper bound obtained from lemma~\ref{singleton} (generalized Singleton bound) when evaluated
for $z-m-k$ adversarial links on the cut $Q$ after erasing $W$ and
$F$ from the graph $\mathcal{G}$. Let
\[N_{z,k,m}(Q)=\min_{\{F\subset Q,|F|=k\leq z-m\}}\min_{\{W\subset Q^R,|W|=m\leq z\}}N_{z,k,m}^{F,W}(Q).\]
Then we define $N_z(Q)$ as follows.
\[N_z(Q)=\min_{0\leq m\leq z}\min_{0\leq k\leq z-m} N_{z,k,m}(Q).\]

For instance, consider the 2-layer zig-zag network in
Fig.~\ref{fig65}. If $z=4$, $k=1$ and $m=1$, $N_{z,k,m}(Q) = 19$ by choosing
$F = \{l_1\}$, $W = \{l_6\}$, and removing $\{l_2,l_3,l_4,l_5\}$ in the application of the Singleton bound after erasing $F$ and $W$.
By taking the minimum over $k$ and $m$, we can show that
$N_z(Q)=19$.

\begin{lemma}(Cut-set upper bound 1)\label{bound1}
Consider any $z$-error correcting network code with source alphabet
$X$ in an acyclic network.
\[\log|X|\leq \min_{Q\in CS(s,t)}\{N_{z}(Q)\}\cdot\log q\]
\end{lemma}
\begin{proof}
For any cut $Q\in CS(s,t)$, the adversary can erase a set $W\subset
Q^R$ of feedback links and a set $F\subset Q$ of forward links where
$|W|=m\leq z$ and $|F|=k\leq z-m$. Applying Lemma \ref{singleton} on
$Q$ after erasing $W$ and $F$ gives the upper bound
$N_{z,k,m}^{F,W}(Q)$. By taking the minimum over all cuts $Q$, we
obtain the above bound.
\end{proof}

The following examples illustrate how the bound in
Lemma~\ref{bound1} tightens the generalized Singleton bound. We
first consider a four-node acyclic network as shown in
Fig.~\ref{fig64}. In each example, unbounded reliable communication
is allowed from source $s$ to its neighbor $B$ on one side of the
cut and from node $A$ to sink $t$ on the other side of the cut.
There are feedback links with arbitrary capacities from $A$ to $B$.

\begin{figure*}
\begin{center}
\begin{picture}(200,200)(-40,-20)\centering
\put(-30,150){\circle{7}}\put(-30,15){\circle{7}}\put(-95,82){\circle{7}}\put(35,82){\circle{7}}
\qbezier(-30,150)(-35,100)(-95,82)
\qbezier(-30,150)(-85,140)(-95,82)
\qbezier(-30,150)(15,130)(35,82)\qbezier(-30,150)(-25,100)(35,82)
\qbezier(-95,82)(-35,65)(-30,15) \qbezier(-95,82)(-85,25)(-30,15)
\qbezier(35,82)(-5,65)(-30,15)\qbezier(35,82)(-25,65)(-30,15)
\qbezier(35,82)(25,5)(-30,15)\qbezier(35,82)(15,35)(-30,15)
\qbezier(-95,82)(-30,100)(35,82)\qbezier(-95,82)(-25,65)(35,82)
\put(-80,130){10}\put(-52,105){10}
\put(-20,60){1}\put(-10,50){1}\put(0,40){1}\put(10,30){1}
\put(-30,89){\circle{1}}\put(-30,82){\circle{1}}\put(-30,75){\circle{1}}
\put(-30,160){$s$}\put(-30,3){$t$}\put(-105,90){$A$}
\put(35,90){$B$}\put(10,135){$\infty$}\put(-90,40){$\infty$}
\put(-95,150){\line(26,-27){130}}\put(40,10){$Q$}\put(115,150){\line(26,-27){130}}\put(250,10){$Q$}
\put(180,150){\circle{7}}\put(180,15){\circle{7}}\put(115,82){\circle{7}}\put(245,82){\circle{7}}
\qbezier(180,150)(155,100)(115,82)\qbezier(180,150)(175,100)(115,82)
\qbezier(180,150)(125,160)(115,82)\qbezier(180,150)(135,130)(115,82)
\qbezier(180,150)(185,100)(245,82)
\qbezier(180,150)(225,130)(245,82)
\qbezier(115,82)(175,65)(180,15) 
\qbezier(115,82)(135,35)(180,15)
\qbezier(245,82)(205,65)(180,15)\qbezier(245,82)(185,65)(180,15)
\qbezier(245,82)(235,5)(180,15)\qbezier(245,82)(225,35)(180,15)
\qbezier(115,82)(180,100)(245,82)\qbezier(115,82)(185,65)(245,82)
\put(180,150){\line(-65,-68){65}}\put(180,15){\line(65,67){65}}
\put(130,130){3}\put(140,122){3}\put(146,115){3}\put(152,110){3}\put(158,105){3}
\put(193,57){2}\put(199,51){2}\put(205,45){1}\put(213,37){1}\put(223,27){1}
\put(180,89){\circle{1}}\put(180,82){\circle{1}}\put(180,75){\circle{1}}
\put(180,160){$s$}\put(180,3){$t$}\put(105,90){$A$}
\put(245,90){$B$}\put(220,135){$\infty$}\put(120,40){$\infty$}
\put(-47,144){\vector(-2,-1){0}}\put(-34,133){\vector(-1,-2){0}}
\put(-13,141){\vector(2,-1){0}}\put(-25,133){\vector(1,-2){0}}
\put(-78,75){\vector(2,-1){0}}\put(-90,66){\vector(1,-2){0}}
\put(-60,89){\vector(2,1){0}}\put(-60,74){\vector(2,-1){0}}
\put(-2,64){\vector(-2,-1){0}}\put(2,60){\vector(-2,-1){0}}
\put(18,52){\vector(-1,-1){0}}\put(27,52){\vector(-1,-1){0}}
\put(158,149){\vector(-2,-1){0}}\put(176,133){\vector(-1,-2){0}}
\put(162,140){\vector(-2,-1){0}}\put(165,135){\vector(-1,-1){0}}\put(171,135){\vector(-1,-1){0}}
\put(205,62){\vector(-1,-1){0}}\put(212,58){\vector(-1,-1){0}}\put(218,54){\vector(-1,-1){0}}
\put(225,48){\vector(-1,-1){0}}\put(236,45){\vector(-1,-1){0}}
\put(197,141){\vector(2,-1){0}}\put(185,133){\vector(1,-2){0}}
\put(132,75){\vector(2,-1){0}}\put(122,66){\vector(1,-2){0}}
\put(150,90){\vector(2,1){0}}\put(150,74){\vector(2,-1){0}}
\put(-35,-25){$(a)$}\put(175,-25){$(b)$}
\end{picture}
\end{center}
\caption{Four-node acyclic network: unbounded reliable communication
is allowed from source $s$ to its neighbor $B$ on one side of the
cut and from node $A$ to sink $t$ on the other side of the cut,
respectively. (a) There are 2 links of capacity 10 from $s$ to
$A$ and 4 unit-capacity links from $B$ to $t$. (b) There are 5 links
of capacity 3 from $s$ to $A$. There are 2 links of capacity
2 and 3 links of capacity 1 from $B$ to $t$.}\label{fig64}
\end{figure*}
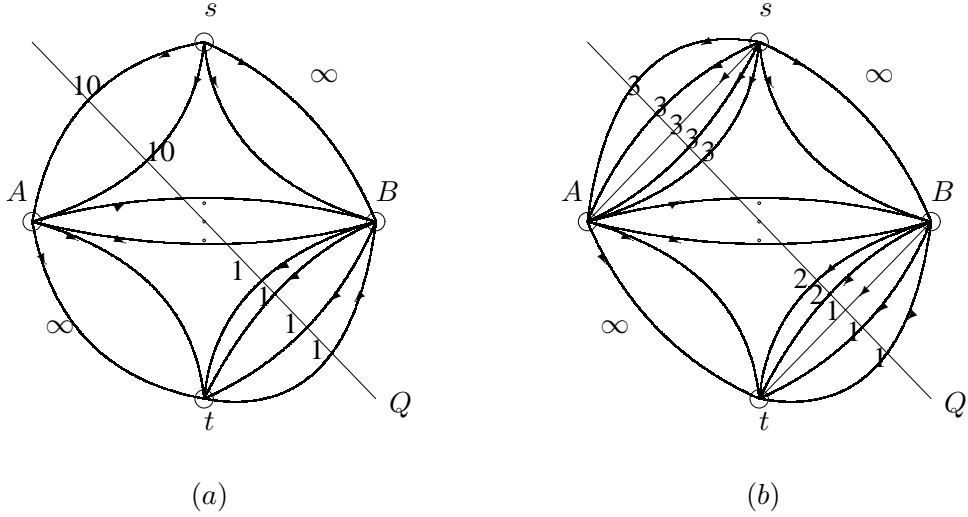

When we compute the generalized Singleton bound, for any cut $Q$, we
choose and erase $2z$ links in the cut such that none of the
remaining links in the cut are downstream of the chosen $2z$ links.
Then we sum the remaining link capacities and take the minimum over
all cuts. Because of the downstream condition, when the link
capacities between $s$ and $A$ are much larger than the link
capacities between $B$ and $t$, the Singleton bound may not be
tight. For example, in the network in Fig.~\ref{fig64} (a), if
$z=2$, then the generalized Singleton bound gives upper bound 20.
However, when the adversary declares that he will use two forward
links between $s$ and $A$, we obtain the erasure bound 4.

As another example, consider the network in Fig.~\ref{fig64} (b)
when $z=2$. Applying the generalized Singleton bound gives upper
bound 16. If the adversary erases one of the forward links between
$s$ and $A$ and we apply the generalized Singleton bound on the
remaining network, then our upper bound is improved to 15. The
intuition behind this example is that when the adversary erases
$k\leq z$ large capacities links which do not satisfy the downstream
condition, applying the generalized Singleton bound on remaining
network with $(z-k)$ adversarial links can give a tighter bound.

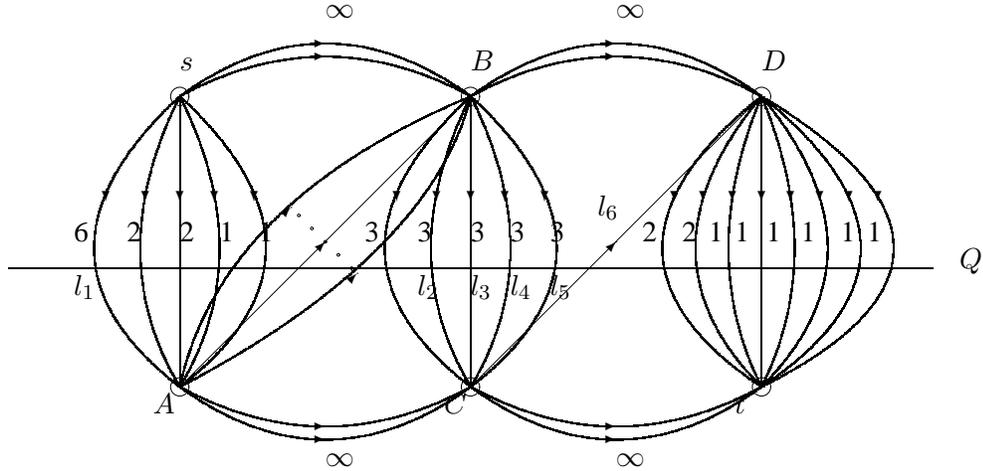
\begin{figure*}
\begin{center}
\begin{picture}(200,200)(-40,-20)\centering
\put(75,20){\circle{7}}\put(75,130){\circle{7}}\put(-35,20){\circle{7}}\put(-35,130){\circle{7}}
\put(185,20){\circle{7}}\put(185,130){\circle{7}}
\qbezier(75,130)(45,70)(75,20)\qbezier(75,130)(10,70)(75,20)
\qbezier(75,130)(105,70)(75,20)\qbezier(75,130)(140,70)(75,20)
\qbezier(-35,130)(-65,70)(-35,20)\qbezier(-35,130)(-100,70)(-35,20)
\qbezier(-35,130)(-5,70)(-35,20)\qbezier(-35,130)(30,70)(-35,20)
\qbezier(185,130)(160,70)(185,20)\qbezier(185,130)(110,70)(185,20)\qbezier(185,130)(135,70)(185,20)\qbezier(185,130)(285,70)(185,20)
\qbezier(185,130)(210,70)(185,20)\qbezier(185,130)(235,70)(185,20)\qbezier(185,130)(260,70)(185,20)
\qbezier(75,20)(130,-20)(185,20)\qbezier(75,20)(130,-10)(185,20)
\qbezier(-35,20)(20,-20)(75,20)\qbezier(-35,20)(20,-10)(75,20)
\qbezier(-35,130)(20,170)(75,130)\qbezier(-35,130)(20,160)(75,130)
\qbezier(75,130)(130,170)(185,130)\qbezier(75,130)(130,160)(185,130)
\qbezier(-35,20)(-20,90)(75,130)\qbezier(-35,20)(60,70)(75,130)
\put(75,130){\line(0,-1){110}}\put(185,130){\line(0,-1){110}}\put(-35,130){\line(0,-1){110}}
\put(-35,20){\line(1,1){110}}\put(75,20){\line(1,1){110}}
\put(20,75){\circle{1}}\put(10,85){\circle{1}}\put(30,65){\circle{1}}\put(15,80){\circle{1}}\put(25,70){\circle{1}}
\put(-35,140){$s$}\put(-45,10){$A$}\put(75,140){$B$}\put(65,10){$C$}\put(185,140){$D$}\put(175,10){$t$}
\put(-75,75){6}\put(-55,75){2}\put(-35,75){2}\put(-20,75){1}\put(-5,75){1}
\put(35,75){3}\put(55,75){3}\put(75,75){3}\put(90,75){3}\put(105,75){3}
\put(140,75){2}\put(155,75){2}\put(165,75){1}\put(175,75){1}\put(187,75){1}\put(200,75){1}\put(215,75){1}\put(225,75){1}
\put(-35,90){\vector(0,-1){0}}\put(-48,90){\vector(0,-1){0}}\put(-22,90){\vector(0,-1){0}}
\put(-7,90){\vector(0,-1){0}}\put(-63,90){\vector(0,-1){0}}
\put(75,90){\vector(0,-1){0}}\put(62,90){\vector(0,-1){0}}\put(88,90){\vector(0,-1){0}}
\put(103,90){\vector(0,-1){0}}\put(47,90){\vector(0,-1){0}}
\put(185,90){\vector(0,-1){0}}\put(174,90){\vector(0,-1){0}}\put(163,90){\vector(0,-1){0}}
\put(196,90){\vector(0,-1){0}}\put(207,90){\vector(0,-1){0}}\put(218,90){\vector(0,-1){0}}
\put(152,90){\vector(0,-1){0}}\put(229,90){\vector(0,-1){0}}
\put(130,75){\vector(1,1){0}}\put(20,75){\vector(1,1){0}}
\put(7,88){\vector(1,1){0}}\put(32,63){\vector(1,1){0}}
\put(20,150){\vector(1,0){0}}\put(20,0){\vector(1,0){0}}
\put(20,145){\vector(1,0){0}}\put(20,5){\vector(1,0){0}}
\put(130,150){\vector(1,0){0}}\put(130,0){\vector(1,0){0}}
\put(130,145){\vector(1,0){0}}\put(130,5){\vector(1,0){0}}
\put(20,160){$\infty$}\put(20,-10){$\infty$}\put(130,160){$\infty$}\put(130,-10){$\infty$}
\put(-100,65){\line(1,0){350}}\put(260,65){$Q$}\put(-75,55){$l_1$}
\put(55,55){$l_2$}\put(75,55){$l_3$}\put(90,55){$l_4$}\put(105,55){$l_5$}\put(123,85){$l_6$}
\end{picture}
\end{center}
\caption{2-layer zig-zag network: unbounded reliable communication
is allowed from $s$ to $B$, from $B$ to $D$, from $A$ to $C$, and
from $C$ to $t$ respectively. There are sufficiently large number of
feedback links from $A$ to $B$. There is one feedback link from $C$
to $D$.}\label{fig65}
\end{figure*}

For the 2-layer zig-zag network in Fig.~\ref{fig65}, when $z=4$, the
min-cut is 37 and the generalized Singleton bound gives upper bound
27. Suppose that the adversary declares that he will use the
feedback link between $C$ and $D$ and the forward link with capacity
6 between $s$ and $A$. By applying the generalized Singleton bound
on the remaining network with two adversarial links, we obtain
37-6-(3+3+3+3)=19. The intuition behind this example is that the
links between $B$ and $C$ and the links between $D$ and $t$ have the
same topological order once the single feedback link between $C$ and
$D$ is erased. Since the generalized Singleton bound is obtained by
erasing $2z$ links on the cut such that none of the remaining links
on the cut is downstream of any erased links, erasing the single
feedback link between $C$ and $D$ yields a tighter Singleton bound
even with fewer adversarial links. Moreover, before applying the
Singleton bound, we first erase the link with capacity 6, which is
the largest link between $s$ and $A$ as we did in example in
Fig.~\ref{fig64}(b).


Next, we introduce another upper bounding approach which considers confusion between two possible sets of $z$ adversarial links, each containing some forward links as well as the corresponding downstream feedback links required to prevent error propagation. Consider any cut $Q=cut(P,\mathcal{V}\backslash P)$ and sets
$Z_1,Z_2\subset Q$. We say that a  feedback link $l\in Q^R$ is directly downstream  of a forward link $l'\in Q$ (and that $l'$ is directly upstream of $l$) if there is a directed path starting from  $l'$ and ending with $l$ that does not include other links in $Q$ or $Q^R$. Let $W_1$  be the set of links in
$Q^R$ which are directly downstream of a link in $Z_1$ and upstream of a link
in $(Q\backslash Z_1)\backslash Z_2$.  Let $W_2$  be the set of links in
$Q^R$ which are directly downstream of a link in $Z_2$ and upstream of a link
in $Q\backslash Z_2$.

\begin{lemma}(Cut-set upper bound 2)\label{bound2}
Let $M=\sum_{(a,b)\in (Q\backslash Z_1)\backslash Z_2}r(a,b)$ denote
the total capacity of the remaining links on $(Q\backslash$
$Z_1)\backslash Z_2$. If $|Z_i\cup W_i|\leq z$ for $i=1,2$, then
\[\log|X|\leq M\cdot\log q.\]
\end{lemma}

\begin{proof}
We assume that $|X|> q^{M}$, and show that this leads to a
contradiction. Let $K(Q)$ denote the number of links on the cut $Q$.
Since $|X|>q^{M}$, from the definition of $M$, there exist two
distinct codewords $x,x'\in X$ such that error-free outputs on the
links in $(Q\backslash Z_1)\backslash Z_2$ are the same. Let
$c=|Z_1|$ and $d=|Z_2|$.  Then we can write
\[O(x) = \{y_1,..,y_{K(Q)-c-d},u_{1},..,u_{c},w_1,..,w_{d}\},\]
\[O(x') = \{y_1,..,y_{K(Q)-c-d},u_{1}',..,u_{c}',w_1',..,w_{d}'\},\]
where $(y_1,..,y_{K(Q)-c-d})$ denotes the error-free outputs on the
links in $(Q\backslash Z_1)\backslash Z_2$ for $x$ and $x'$; $(u_{1},..,u_{c})$ and  $(u_{1}',..,u_{c}')$
denote the error-free outputs on the links in $Z_1$ for $x$ and $x'$ respectively; and $(w_{1},..,w_{d})$ and
$(w_{1}',..,w_{d}')$ denote the error-free outputs on the links in
$Z_2$ for $x$ and $x'$ respectively. We will show that it is possible for the adversary
to produce exactly the same outputs on all the channels in $Q$ under  $x$ and $x'$ when
errors occur on at most $z$ links. When codeword $x$ is sent, we use
$B_l(x)$ to denote the error-free symbols on feedback link $l$.

Assume the input of network is $x$. The adversary chooses feedback
links set $W_1$ and forward links set $Z_1$ as its $z$ adversarial
links. First the adversary applies errors on $Z_1$ to change the
output from $u_{i}$ to $u_{i}'$ for $\forall 1\leq i\leq c$ and to cause
each feedback link $l\in W_1$ to transmit $B_l(x)$. Since all feedback
links which are directly downstream of a link in $Z_1$ and upstream of a
link in $(Q\backslash Z_1)\backslash Z_2$ transmit the error-free
symbols,
the outputs on links in $(Q\backslash Z_1)\backslash Z_2$  are not affected. The outputs  on links in $Z_2$ may be affected, and we denote their new values by $\{w''_1,..,w''_d\}$. Thus, the sink observes
$\{y_1,..,y_{K(Q)-c-d},u_{1}',..,u_c',w''_1,..,w''_d\}$.

When codeword $x'$ is transmitted, the adversary chooses feedback
links set $W_2$ and forward links set $Z_2$ as its $z$ adversarial
links. The adversary applies errors on them to change
$(w_1,..,w_d)$ to $(w_1'',..,w_d'')$ and to cause each feedback link $l\in W_2$ to transmit $B_l(x')$. Since all feedback
links which are directly downstream of a link in $Z_2$ and upstream of a
link in $Q\backslash Z_2$ transmit the error-free
symbols, the outputs on
any other links are not affected.
Therefore, the output is changed from $O(x')$
to $\{y_1,..,y_{K(Q)-c-d},u_{1}',..,u_c',w''_1,..,w''_d\}$. Thus, the
sink node $t$ cannot reliably distinguish between the codewords $x$
and $x'$, which gives a contradiction.
\end{proof}

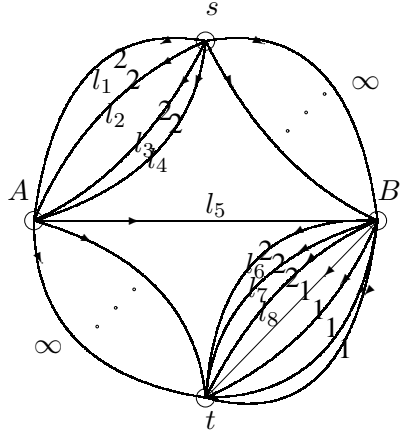
\begin{figure}
\begin{center}
\begin{picture}(200,200)(-40,-20)\centering
\put(75,150){\circle{7}}\put(75,15){\circle{7}}\put(10,82){\circle{7}}\put(140,82){\circle{7}}
\qbezier(75,150)(50,100)(10,82)\qbezier(75,150)(70,100)(10,82)
\qbezier(75,150)(20,160)(10,82)\qbezier(75,150)(30,130)(10,82)
\qbezier(75,150)(100,100)(140,82)
\qbezier(75,150)(130,160)(140,82)
\qbezier(10,82)(70,65)(75,15)
\qbezier(10,82)(10,25)(75,15)
\qbezier(140,82)(80,85)(75,15)\qbezier(140,82)(100,65)(75,15)\qbezier(140,82)(80,65)(75,15)
\qbezier(140,82)(130,15)(75,15)\qbezier(140,82)(120,35)(75,15)
\qbezier(140,82)(130,0)(75,15)
\put(75,15){\line(65,67){65}}\put(10,82){\line(1,0){130}}
\put(75,160){$s$}\put(75,3){$t$}\put(0,90){$A$} \put(140,90){$B$}
\put(120,130){\circle{1}}\put(113,123){\circle{1}}\put(106,116){\circle{1}}
\put(34,42){\circle{1}}\put(48,56){\circle{1}}\put(41,49){\circle{1}}
\put(40,140){2}\put(45,132){2}\put(57,120){2}\put(61,115){2}
\put(32,132){$l_1$}\put(37,119){$l_2$}\put(48,108){$l_3$}\put(53,102){$l_4$}
\put(75,85){$l_5$}\put(90,62){$l_6$}\put(91,53){$l_7$}\put(95,44){$l_8$}
\put(95,67){2}\put(100,62){2}\put(105,57){2}\put(110,52){1}\put(115,45){1}\put(120,38){1}\put(125,30){1}
\put(58,149){\vector(-2,-1){0}}\put(58,141){\vector(-2,-1){0}}
\put(65,133){\vector(-1,-1){0}}\put(71,133){\vector(-1,-2){0}}
\put(97,149){\vector(2,-1){0}}\put(85,133){\vector(1,-2){0}}
\put(32,74){\vector(2,-1){0}}\put(13,66){\vector(1,-2){0}}
\put(50,82){\vector(1,0){0}}
\put(110,70){\vector(-2,-1){0}}\put(108,77){\vector(-2,-1){0}}\put(113,66){\vector(-1,-1){0}}
\put(120,61){\vector(-1,-1){0}}\put(126,58){\vector(-1,-1){0}}\put(132,54){\vector(-1,-1){0}}
\put(135,54){\vector(-1,-1){0}}
\put(10,32){$\infty$}\put(130,132){$\infty$}
\end{picture}
\end{center}
\caption{Four node acyclic network: There are 4 links of capacity 2 from $s$ to $A$. There are 3 links of capacity 2 and
1 links of capacity 4 from $B$ to $t$.}\label{fig66}
\end{figure}

Given a cut $Q$, we consider all possible sets $(Z_1,Z_2)$ on $Q$
satisfying the condition of Lemma \ref{bound2}. We choose sets
$(Z_1^*,Z_2^*)$ among them that have the maximum total link
capacities and define $M_z(Q)$ to be the sum of the capacities of
the links in $(Q\backslash Z_1^*)\backslash Z_2^*$. This gives the
upper bound
\[\log|X|\leq \min_{Q\in cut(s,t)}M_z(Q)\cdot\log
q.\]

The following example shows that we can obtain a tighter upper bound
using Lemma~\ref{bound2}. For the example network in
Fig.~\ref{fig66}, when $z=3$, Lemma~\ref{bound1} gives upper bound
9. However, Lemma~\ref{bound2} gives a tighter upper bound 8 when
$Z_1^*=\{l_1,l_2,l_5\}$ , and $Z_2^*=\{l_6,l_7,l_8\}$.

Now we derive a generalized cut-set upper bound that unifies
Lemma~\ref{bound1} and Lemma~\ref{bound2}. Given a cut $Q$, consider a set $F\subset Q$ of forward links and a set $W\subset Q^R$ of feedback links such that $|F|+|W|\le z$. Let
$C_z^{F,W}(Q)$ denote the upper bound obtained from Lemma
\ref{bound2} when evaluated for $z-m-k$ adversarial links on the cut
$Q$ after erasing $F$ and  $W$ from the original graph $\mathcal{G}$. Then
\[\min_{\{F\subset Q,W\subset Q^R,|F|+|W|\le z\}}C_z^{F,W}(Q)\]is an upper bound on the error correction capacity of $\mathcal{G}$. This includes the bound $N_z (Q)$ of Lemma~\ref{bound1} as a special case, since the generalized Singleton bound is a special case of the upper bound in
Lemma~\ref{bound2} corresponding to the case where $Z_1\cup Z_2$ is a set of $2z$ links satisfying the downstream
condition. It is also clear that
$C_z^{F,W}(Q)$ is the same as the bound in Lemma~\ref{bound2} when
$F=W=\emptyset$.  Note however that any bound $C_z^{F,W}(Q)$ obtainable with a nonempty set $W$ of erased feedback links is also obtainable by including  those links in the sets $W_1 $ and $W_2$ of Lemma~\ref{bound2} instead of erasing them. Thus,  we define
\[C_z(Q)=\min_{\{F\subset Q,|F|\leq z\}}C_z^{F,\emptyset}(Q)\]
and state our upper bound as follows.\footnote{
After submitting this paper, we found a way to tighten the above
bound for the case of small feedback link capacity. Briefly, the key
idea is to note that instead of choosing all the links in $W_i$ as
adversarial links as in the proof of Lemma~\ref{bound2}, another
possibility is to choose only a subset $Y_i\subset  W_i$ as
adversarial links, as long as the values on links in $W_i\backslash
Y_i $ and links in $Q\backslash Z_i$ that are directly upstream of
links in $W_i\backslash Y_i $ are the same under the two confusable
codewords $x$ and $x'$. The capacities of these links  then appear
as part of the upper  bound; thus, this bound is useful for cases
where feedback links have small capacity. This result will be
presented formally in an upcoming paper.}
\begin{theorem}\label{bound}
(A generalized cut-set upper bound)
Consider any $z$-error
correcting network code with source alphabet $X$ in an acyclic
network. Then\[\log|X|\leq \min_{Q\in CS(s,t)}C_{z}(Q)\cdot\log q.\]
\end{theorem}

\section{Coding strategies}\label{section:strategies}

We consider a variety of linear and nonlinear coding strategies
useful for achieving the capacity of various example networks. We show
the insufficiency of linear network codes for achieving the capacity
in general. We also demonstrate
examples of networks with a single source and a single sink where,
unlike the equal link capacity case, it is necessary for
intermediate nodes to do coding, nonlinear error detection or error
correction in order to achieve the capacity. We then introduce a new
coding strategy, guess-and-forward.

\subsection{Error detection at intermediate nodes and insufficiency of linear codes}

Here we show that there exists a network where the capacity is 50\%
greater than the best rate that can be achieved with linear coding.
We consider the single source and the single sink network in
Fig.~\ref{fig67}, where source $s$ aims to transmit the information
to a sink node $t$. We index the links and assume the capacities of
links as shown in Fig.~\ref{fig67}. For a single adversarial link,
our upper bound from Theorem~\ref{bound} is 2.

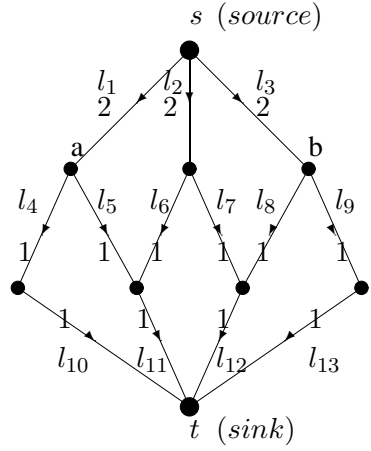
\begin{figure}
\begin{center}
\begin{picture}(200,200)(-40,-20)\centering
\put(75,150){\circle*{7}}\put(30,105){\circle*{5}}\put(75,105){\circle*{5}}\put(120,105){\circle*{5}}
\put(10,60){\circle*{5}}\put(55,60){\circle*{5}}\put(95,60){\circle*{5}}\put(140,60){\circle*{5}}\put(75,15){\circle*{7}}
\put(75,150){\line(0,-1){45}}
\put(75,150){\line(-1,-1){45}}\put(75,150){\line(1,-1){45}}
\put(75,105){\line(-4,-9){20}}\put(75,105){\line(4,-9){20}}
\put(30,105){\line(-4,-9){20}}\put(30,105){\line(5,-9){25}}\put(120,105){\line(4,-9){20}}
\put(120,105){\line(-5,-9){25}}
\put(10,60){\line(13,-9){65}}\put(55,60){\line(4,-9){20}}\put(95,60){\line(-4,-9){20}}\put(140,60){\line(-13,-9){65}}
\put(75,160){$s$}\put(75,3){$t$}\put(40,135){$l_1$}\put(65,135){$l_2$}\put(100,135){$l_3$}\put(10,90){$l_4$}\put(40,90){$l_5$}
\put(60,90){$l_6$}\put(85,90){$l_7$}\put(100,90){$l_8$}\put(130,90){$l_9$}\put(25,30){$l_{10}$}\put(55,30){$l_{11}$}
\put(85,30){$l_{12}$}\put(120,30){$l_{13}$}
\put(85,160){$(source)$}\put(85,3){$(sink)$}
\put(40,125){$2$}\put(65,125){$2$}\put(100,125){$2$}\put(10,70){$1$}\put(40,70){$1$}
\put(60,70){$1$}\put(85,70){$1$}\put(100,70){$1$}\put(130,70){$1$}\put(25,45){$1$}\put(55,45){$1$}
\put(85,45){$1$}\put(120,45){$1$} \put(30,110){a}\put(120,110){b}
\put(75,130){\vector(0,-1){0}} \put(55,130){\vector(-1,-1){0}}
\put(95,130){\vector(1,-1){0}} \put(19,80){\vector(-1,-2){0}}
\put(44,80){\vector(1,-2){0}}
\put(64,80){\vector(-1,-2){0}}\put(87,80){\vector(1,-2){0}}
\put(105,80){\vector(-1,-2){0}}\put(130,80){\vector(1,-2){0}}
\put(39,40){\vector(1,-1){0}}\put(64,40){\vector(1,-2){0}}
\put(86,40){\vector(-1,-2){0}}\put(111,40){\vector(-1,-1){0}}
\end{picture}
\end{center}
\caption{A single source and a single sink network : all links on
the top layer have capacity 2. All links on the middle and bottom
layer have capacity 1. When $z=1$, the capacity of this network is 2
while linear network codes achieve at most 4/3.}\label{fig67}
\end{figure}

\begin{lemma}\label{insuffex}
Given a network in Fig.~\ref{fig67}, for a single adversarial link,
rate 2 is asymptotically achievable with a nonlinear error detection
strategy, whereas scalar linear network coding achieves at most 4/3.
\end{lemma}

\begin{proof}
We first illustrate the nonlinear error detection strategy as
follows. Source wants to transmit two packets $(P_1,P_2)$. We send them
in $n$ channel uses, but each packet has only $n-1$ bits. We use one
bit as a signaling bit. We send $(P_1,P_2)$ down all links in the top
layer. In the middle layer, we do the following operations:

1. Send the linear combination of $P_1$ and $P_2$, $aP_1+bP_2$, down link
$l_4$.

2. Send $P_1$ down both links $l_5$ and $l_6$.

3. Send $P_2$ down both links $l_7$ and $l_8$.

4. Send a different linear combination of $P_1$ and $P_2$, $cP_1+dP_2$,
down link $l_9$.

At the bottom layer, we do the following operations:

1. Forward the received packet on link $l_{10}$.

2. Send a 1 followed by $P_1$ on link $l_{11}$ if the two copies of
$P_1$ match, send a 0 otherwise.

3. Send a 1 followed by $P_2$ on link $l_{12}$ if the two copies of
$P_2$ match, send a 0 otherwise.

4. Forward the received packet on link $l_{13}$.

We can show that above nonlinear error detection strategy allows a
sink node to decode ($P_1,P_2$). Suppose that $(a,b)$ and $(c,d)$ are
independent. Then coding vectors on any two links on the bottom
layer are independent and they satisfy with MDS (maximum distance
separable) properties. If nothing was sent down both $l_{11}$ and
$l_{12}$, the decoder can recover $(P_1,P_2)$ from the information
received on links $l_{10}$ and $l_{13}$. If nothing was sent down
only on $l_{11}$, then the outputs of $l_{12}$ and $l_{13}$ should
not be corrupted and the decoder can recover $(P_1,P_2)$. Similarly, the
decoder can decode correctly when nothing was sent down only on
$l_{12}$. If all the links in the bottom layer received symbols,
there is at most one erroneous link on the bottom layer, which has
MDS code. Thus we can achieve rate $2-\frac{2}{n}$ with this error
detection strategy.

Now we show that scalar linear network code can achieve at most rate
4/3. Suppose that we want to achieve the linear coding capacity $k/n$ by
transmitting $k$ symbols reliably by using a scalar linear network
code $\phi$ over the finite field $\mathbb{F}_q$ in $n$ rounds. To show the insufficiency of
linear coding for achieving this capacity, from (\ref{eq1}), it is
sufficient to prove that there exist pairs $(w,e)$ and $(w',e')$ for
linear network code $\phi$ such that
\[(\psi_l(w,e):l\in \Gamma_+(t))=(\psi_l(w',e'):l\in \Gamma_+(t)),\]
and $N(e), N(e')\leq 1$. Since the above equation is equivalent to
\[(\tilde{\phi}_l(w-w'):l\in \Gamma_+(t))=(\theta_l(-e+e'):l\in \Gamma_+(t)),\]
by linearity, it suffices to find a source vector $x\in X$ and error vector $e''$ such that
$N(e'')\leq 2$ and
\begin{equation}\label{eq2}
(\tilde{\phi}_l(x):l\in \Gamma_+(t))=(\theta_l(e''):l\in
\Gamma_+(t)),
\end{equation}
where $X=\mathbb{F}_q^k$ is the source alphabet. We will show
that there exists $(x,e'')$ satisfying the above equation when errors
occur on the links $l_1$ and $l_3$ in error vector $e''$.

Let $G_t$ denote the $4n\times k$  transfer matrix
between $s$ and $t$ in the $n$ rounds. Its rows are the global
coding vectors assigned on $l_{10}$, $l_{11}$, $l_{12}$, and
$l_{13}$ in the $n$ rounds. Note that to transmit $k$ symbols reliably, $G_t$ should have rank $k$.

Let $M_1$ and $M_2$ denote the transfer matrices between $l_1$ and $t$,
and between $l_3$ and $t$ during $n$ rounds respectively. To
transmit $k$ symbols reliably, both $M_1$ and $M_2$
should have rank at least $k$, i.e., $ rank(M_1)\geq k$ and
$rank(M_2)\geq k$. Otherwise, when the adversarial link is on the top
layer, the maximum achievable rate is at most
$\min\{rank(M_1),rank(M_2)\}$ from the data processing inequality,
which gives a contradiction.

Let $e_1$ and $e_2$ denote the errors occurring on links $l_1$ and
$l_3$, respectively. Error $e_1$ propagates to $l_{10}$ and
$l_{11}$, and error $e_2$ propagates to $l_{12}$ and $l_{13}$.

From (2), we have the following set of equations
\[ G_t x =\left( \begin{array}{cc}
M_1 & 0  \\
0 & M_2 \end{array} \right)(e_1,e_2)^{\tau} = M\cdot e''.\]
%

Since $rank(M_1)\geq k$ and $rank(M_2)\geq
k$, $rank(M)\geq 2k$. Then $A=\{G_tx:x\in X\}$ and $B=\{Me'':e''\in
\mathbb{F}^{4n}_q\}$ are both linear subspaces of $\mathbb{F}^{4n}_q$, and
$\dim(A) = k$ and $\dim(B)\geq 2k$.

Let $\{x_1,..,x_{k}\}$ denote a basis of $X$. Then
$\{G_tx_1,..,G_tx_{k}\}$ is a basis of $A$. Similarly, since
$rank(M)\geq 2k$, there exist $2k$ vectors $\{y_1,..,y_{2k}\}$ such
that $\{My_1,..,My_{2k}\}$ is a subset of basis of $B$.

If $3k>4n$, since both $A$ and $B$ are linear subspaces of
$\mathbb{F}^{4n}_q$, there exists $(a_1,..,a_{k},b_1,..,b_{2k})\neq
(0,...,0)$ such that
\[\sum_{i=1}^{k}a_i(G_tx_i)+\sum_{j=1}^{2k}b_i(My_i)=0.\] If
$(a_1,..,a_{k})= (0,...,0)$ or $(b_1,..,b_{2k})=(0,...,0)$, then it
contradicts the linear independence of basis. Thus,
$(a_1,..,a_{k})\neq (0,...,0)$ and $(b_1,..,b_{2k})\neq 0$. Then,
\begin{eqnarray*}
\lefteqn{\sum_{i=1}^{k}a_i(G_tx_i)+\sum_{j=1}^{2k}b_i(My_i)}\\
& = & \sum_{i=1}^{k}G_t(a_ix_i)+\sum_{j=1}^{2k}M(b_iy_i)\\
& = &\sum_{i=1}^{k}G_t(a_ix_i)-\sum_{j=1}^{2k}M(-b_iy_i)\\
& = & 0.\\
\end{eqnarray*}
Therefore, we have found nonzero $x=\sum_{i=1}^{k}a_ix_i$ and
$(e_1,e_2)^{\tau} = -\sum_{j=1}^{2k}(-b_jy_j)$ such that $G_tx =
Me''$.

It completes the proof.

\end{proof}

The following corollary shows that vector linear
codes\footnote{A vector linear code is a
linear code operating over a vector of symbols.} (see e.g.~\cite{yeung2008information}) also achieve at
most rate 4/3.

\begin{corollary}
For the network in Fig.~\ref{fig67} with a single adversarial link,
a vector linear network code can achieve at most rate 4/3.
\end{corollary}

\begin{proof}
For a vector linear code, the outgoing edges of
each node carries vectors of alphabet symbols which are function of
the vectors carried on the incoming edges to the node.
We consider a vector linear code that groups $m$ symbols into a
vector. As in lemma~\ref{insuffex}, we define the $(4n)m\times km$
transfer matrix $G_t$ between $s$ and $t$. Transfer matrices $M_1$
and $M_2$ are also defined in the same way, and $rank(M_1)\geq km$
and $rank(M_2)\geq km$. As in the proof of lemma~\ref{insuffex},
when $k>\frac{4n}{3}$, we can show that there exists vectors
$(x,e_1,e_2)$ $(x\neq 0)$ satisfying
\[G_t x = (M_1\cdot e_1, M_2\cdot e_2).\]
\end{proof}


\subsection{Error correction at intermediate nodes}
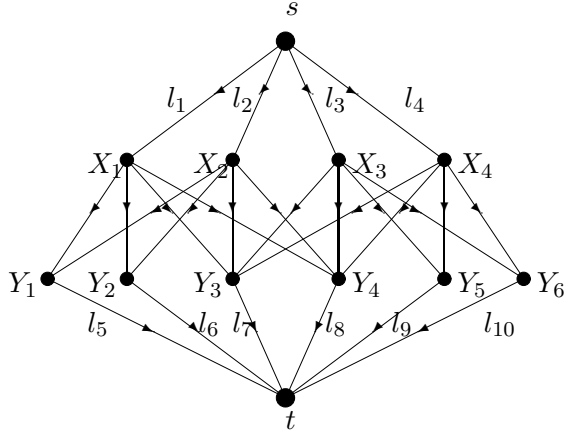
\begin{figure}
\begin{center}
\begin{picture}(200,200)(-40,-20)\centering
\put(75,150){\circle*{7}}\put(15,105){\circle*{5}}\put(55,105){\circle*{5}}\put(95,105){\circle*{5}}\put(135,105){\circle*{5}}
\put(-15,60){\circle*{5}}\put(15,60){\circle*{5}}\put(55,60){\circle*{5}}\put(95,60){\circle*{5}}\put(135,60){\circle*{5}}\put(165,60){\circle*{5}}
\put(75,15){\circle*{7}}
\put(75,150){\line(-4,-3){60}}\put(75,150){\line(-4,-9){20}}\put(75,150){\line(4,-9){20}}\put(75,150){\line(4,-3){60}}
\put(15,105){\line(-2,-3){30}}\put(15,105){\line(0,-1){45}}\put(15,105){\line(8,-9){40}}\put(15,105){\line(16,-9){80}}
\put(55,105){\line(-8,-9){40}}\put(55,105){\line(0,-1){45}}\put(55,105){\line(8,-9){40}}\put(55,105){\line(-14,-9){70}}
\put(95,105){\line(-8,-9){40}}\put(95,105){\line(0,-1){45}}\put(95,105){\line(8,-9){40}}\put(95,105){\line(14,-9){70}}
\put(135,105){\line(-8,-9){40}}\put(135,105){\line(0,-1){45}}\put(135,105){\line(2,-3){30}}\put(135,105){\line(-16,-9){80}}
\put(-15,60){\line(2,-1){90}}\put(15,60){\line(4,-3){60}}\put(55,60){\line(4,-9){20}}
\put(95,60){\line(-4,-9){20}}\put(135,60){\line(-4,-3){60}}\put(165,60){\line(-2,-1){90}}
\put(75,160){$s$}\put(75,3){$t$}
\put(30,125){$l_1$}\put(55,125){$l_2$}\put(90,125){$l_3$}\put(120,125){$l_4$}
\put(0,40){$l_5$}\put(42,40){$l_6$}\put(55,40){$l_7$}\put(90,40){$l_8$}\put(115,40){$l_9$}\put(150,40){$l_{10}$}
\put(-30,55){$Y_1$}\put(0,55){$Y_2$}\put(40,55){$Y_3$}\put(100,55){$Y_4$}\put(140,55){$Y_5$}\put(170,55){$Y_6$}
\put(0,100){$X_1$}\put(40,100){$X_2$}\put(100,100){$X_3$}\put(140,100){$X_4$}
\put(48,130){\vector(-2,-1){0}}\put(65,130){\vector(-1,-1){0}}
\put(102,130){\vector(2,-1){0}}\put(85,130){\vector(1,-1){0}}
\put(1,85){\vector(-1,-1){0}}\put(15,85){\vector(0,-1){0}}
\put(32,85){\vector(1,-1){0}}\put(50,85){\vector(2,-1){0}}
\put(24,85){\vector(-2,-1){0}}\put(38,85){\vector(-1,-1){0}}
\put(55,85){\vector(0,-1){0}}\put(73,85){\vector(1,-1){0}}
\put(126,85){\vector(2,-1){0}}\put(112,85){\vector(1,-1){0}}
\put(95,85){\vector(0,-1){0}}\put(77,85){\vector(-1,-1){0}}
\put(149,85){\vector(1,-1){0}}\put(135,85){\vector(0,-1){0}}
\put(118,85){\vector(-1,-1){0}}\put(100,85){\vector(-2,-1){0}}
\put(25,40){\vector(2,-1){0}}\put(42,40){\vector(1,-1){0}}\put(64,40){\vector(1,-2){0}}
\put(125,40){\vector(-2,-1){0}}\put(108,40){\vector(-1,-1){0}}\put(86,40){\vector(-1,-2){0}}
\end{picture}
\end{center}
\caption{A single source and a single sink network : The link
capacity in this network is as follows:
$r(l_1)=r(l_2)=r(l_3)=r(l_4)=4, r(l_5)=...=r(l_{10})=2$. All the
links in the middle layer have capacity $1$. Error correction at
$Y_3$ and $Y_4$ is necessary for achieving the
capacity.}\label{fig69}
\end{figure}
We next give an example in which error correction at intermediate nodes
is used for achieving the capacity. The intuition behind our
approach is that error correction at intermediate nodes can reduce
the error propagation to the links in the bottom layer and MDS code
assigned on the bottom layer gives the correct output. We consider a
single source-destination network in Fig.~\ref{fig69}. For a single
adversarial link, upper bound from Theorem \ref{bound} is 8. From
Sec. IV-A, the upper bound on the linear coding capacity is
$\sum_{i=5}^{10}r(l_i)/(m+1)=6$.

\begin{lemma}
Given the network in Fig.~\ref{fig69}, for a single adversarial
link, rate 8 is achievable using error correction at intermediate
nodes.
\end{lemma}
\begin{proof}
Without loss of generality, all nodes  except $Y_3$ and $Y_4$ forward their received
information. We first assign a $(12,8)$ MDS
code $(a,b,\ldots,l)$ on the bottom layer links and apply a (4,2) MDS
code at each decision node, e.g., we assign $(e,f,e+f,e+2f)$ and
$(g,h,g+h,g+2h)$ on incoming links to $Y_3$ and $Y_4$ respectively.
Then we can assign codewords on all links in the network since all
nodes except $Y_3$ and $Y_4$ are forwarding nodes. If the adversarial link is
on the middle or bottom layer, at most two errors are propagated to
the sink node and the MDS code assigned on the bottom layer gives the
correct output. If the adversarial link is on the top layer, at most two
errors are propagated to the sink node through forwarding nodes
$Y_1$, $Y_2$, $Y_5$, and $Y_6$. Since at most one error is incoming
to $Y_3$ and $Y_4$ respectively, the (4,2) MDS code applied at each
decision node gives error-free output $(e,f)$ and $(g,h)$.
Therefore, when the adversarial link is on the top layer, at most two
errors are propagated to the sink and the (12,8) MDS code returns the
correct output.

\end{proof}
\begin{algorithm}
\caption{Algorithm for error correction at intermediate nodes}
\label{algeccint}

$M \leftarrow N$

$CS \leftarrow\emptyset$

$\mathcal{G}'=\mathcal{G}$

\hspace{4mm} While {$|M|\geq 1$} and $\exists i$ s.t.
$d(\mathcal{G}',i)>0$, $i\in M$

\hspace{8mm} $M'=M-I(\mathcal{G}',M)$

\hspace{8mm} $CS' = CS\cup S_{I(\mathcal{G}',M)}$

\hspace{8mm} $\mathcal{G}' = \mathcal{G}'-S_{I(\mathcal{G}',M)}$

\hspace{8mm} $M=M'$, $CS=CS'$.

\hspace{4mm} endwhile

\hspace{4mm} return $CS$

\end{algorithm}

One possible generalization of the above intermediate node error
correction is as follows. Given an acyclic network
$\mathcal{G}=(\mathcal{V},\mathcal{E})$, we use
$c_{\mathcal{G}}(s,i)$ and $c_{\mathcal{G}}(i,t)$ denote the min-cut
between the source $s$ and $i$, and the min-cut between $i$ and the
sink $t$ in $\mathcal{G}$, respectively. We assume that there is a
fixed set of nodes $N\subset \mathcal{V}$ such that
$c_{\mathcal{G}}(s,i)\geq c_{\mathcal{G}}(i,t)$ for $\forall i\in N$
and error correction can be applied only at nodes in $N$. For
instance, in Fig.~\ref{fig69}, $\mathcal{N}=\{Y_3,Y_4\}$. Let
$d(\mathcal{G},i)=c_{\mathcal{G}}(s,i)- c_{\mathcal{G}}(i,t)$ denote
the difference between the max-flow from $s$ to $i$ and the max-flow
from $i$ to $t$.

The selection function $I(\mathcal{G},N)$ chooses a node $i\in N$ on
$\mathcal{G}$ which maximizes $d(\mathcal{G},i)$. Precisely,
\[I(\mathcal{G},N) = \arg\max_{i\in N}\{d(\mathcal{G},i)\}.\]

Here is the outline of our greedy algorithm with error correction at
intermediate nodes. Given an acyclic network $\mathcal{G}$ and the
set of error correction nodes $N$, we choose a node
$i=I(\mathcal{G},N)$ that maximizes $d(\mathcal{G},i)$ on
$\mathcal{G}$. Since $c_{\mathcal{G}}(s,i)$ is the max-flow from $s$
to $i$, we can find $c_{\mathcal{G}}(s,i)$ paths so that each path
carries one symbol from $s$ to $i$. Likewise, we also find
$c_{\mathcal{G}}(i,t)$ paths from $i$ to $t$. Let
$S_{I(\mathcal{G},M)}$ denote the subgraph composed of above paths.
We assign a $(c_{\mathcal{G}}(s,i),c_{\mathcal{G}}(i,t))$ MDS code
on $S_{I(\mathcal{G},M)}$. We remove $S_{I(\mathcal{G}',M)}$ from
$\mathcal{G}$ and add it to $CS$ which denotes the union of
subgraphs for which codewords are already assigned. We also remove
$i$ from $N$. We repeat the above procedure until $N=\emptyset$ or
there is no node $i\in N$ such that $d(\mathcal{G},i)> 0$. A precise
description of the algorithm is shown in algorithm~\ref{algeccint}.
Since max-flow can be computed in polynomial-time, this algorithm is
a polynomial-time greedy algorithm.

\subsection{Coding at intermediate nodes}
Here we consider an example of a single-source and single-sink
network, shown in Fig.~\ref{fig68}, whose capacity is achieved using linear coding at
intermediate nodes rather than nonlinear error correction or
detection. For a single adversarial link, the
upper bound obtained from Theorem~\ref{bound} is 4.

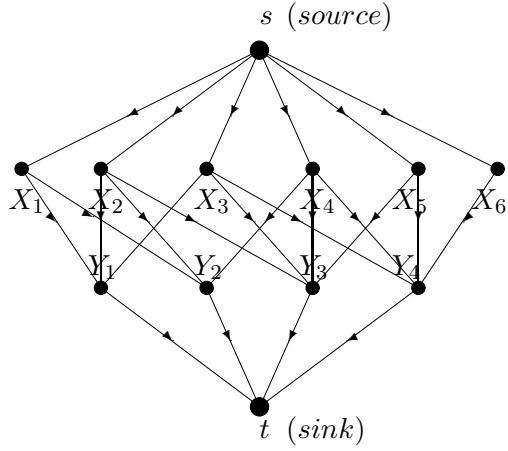
\begin{figure}
\begin{center}
\begin{picture}(200,200)(-40,-20)\centering
\put(75,150){\circle*{7}}\put(-15,105){\circle*{5}}\put(15,105){\circle*{5}}\put(55,105){\circle*{5}}
\put(95,105){\circle*{5}}\put(135,105){\circle*{5}}\put(165,105){\circle*{5}}
\put(15,60){\circle*{5}}\put(55,60){\circle*{5}}\put(95,60){\circle*{5}}\put(135,60){\circle*{5}}
\put(75,15){\circle*{7}}
\put(75,150){\line(-2,-1){90}}\put(75,150){\line(-4,-3){60}}\put(75,150){\line(-4,-9){20}}
\put(75,150){\line(4,-9){20}}\put(75,150){\line(4,-3){60}}\put(75,150){\line(2,-1){90}}
\put(-15,105){\line(2,-3){30}}\put(15,105){\line(0,-1){45}}\put(55,105){\line(-8,-9){40}}
\put(-15,105){\line(14,-9){70}}\put(15,105){\line(8,-9){40}}\put(95,105){\line(-8,-9){40}}
\put(15,105){\line(16,-9){80}}\put(55,105){\line(8,-9){40}}\put(95,105){\line(0,-1){45}}\put(135,105){\line(-8,-9){40}}
\put(55,105){\line(16,-9){80}}\put(95,105){\line(8,-9){40}}\put(135,105){\line(0,-1){45}}\put(165,105){\line(-2,-3){30}}
\put(15,60){\line(4,-3){60}}\put(55,60){\line(4,-9){20}}
\put(95,60){\line(-4,-9){20}}\put(135,60){\line(-4,-3){60}}
\put(75,160){$s$}\put(75,3){$t$}\put(85,160){$(source)$}\put(85,3){$(sink)$}
\put(-20,90){$X_1$}\put(10,90){$X_2$}\put(50,90){$X_3$}\put(90,90){$X_4$}\put(125,90){$X_5$}\put(155,90){$X_6$}
\put(10,65){$Y_1$}\put(50,65){$Y_2$}\put(90,65){$Y_3$}\put(125,65){$Y_4$}
\put(25,125){\vector(-2,-1){0}}\put(42,125){\vector(-1,-1){0}}\put(64,125){\vector(-1,-2){0}}
\put(86,125){\vector(1,-2){0}}\put(125,125){\vector(2,-1){0}}\put(108,125){\vector(1,-1){0}}
\put(-2,85){\vector(1,-1){0}}\put(12,87){\vector(2,-1){0}}
\put(42,40){\vector(1,-1){0}}\put(64,40){\vector(1,-2){0}}
\put(86,40){\vector(-1,-2){0}}\put(108,40){\vector(-2,-1){0}}
\put(73,85){\vector(1,-1){0}}\put(91,85){\vector(2,-1){0}}
\put(15,85){\vector(0,-1){0}}\put(33,85){\vector(1,-1){0}}\put(51,85){\vector(2,-1){0}}
\put(77,85){\vector(-1,-1){0}}\put(95,85){\vector(0,-1){0}}\put(113,85){\vector(1,-1){0}}
\put(117,85){\vector(-1,-1){0}}\put(135,85){\vector(0,-1){0}}\put(152,85){\vector(-1,-1){0}}
\end{picture}
\end{center}
\caption{A single source and a single sink network : all links on
the top or middle layer have capacity one. All links on the bottom
layer have capacity 2. In this network, coding at intermediate nodes
but not error-detection and correction is necessary for achieving
the capacity.}\label{fig68}
\end{figure}

\begin{lemma}
Given the network in Fig.~\ref{fig68}, for a single adversarial
link, coding at intermediate nodes achieves the rate 4.
\end{lemma}

\begin{proof}
To achieve rate 4, any four links on the top layer should carry 4
independent packets. Otherwise, when the adversarial link is on the
top layer, source cannot transmit 4 packets reliably. Then the data
processing inequality gives a contradiction. Similarly, any two
links on the bottom layer should carry 4 independent packets. Since
$Y_i$ is connected to at most four different nodes among
$(X_1,..,X_6)$ for $\forall 1\leq i\leq 4$ and all links in the
middle layer have capacity 1, each of $Y_1$, $Y_2$, $Y_3$, and $Y_4$
receives all independent information. Thus we cannot apply simple
error-detection or correction at $Y_1$, $Y_2$, $Y_3$, and $Y_4$.
Suppose that only forwarding strategy is used on this network. Then
we show that rate 4 is not achievable. There are six symbols on the
top layer. Since we use only forwarding, these are forwarded to the
bottom layer. Since bottom layer links have total capacity 8, there
are at least two same symbols on the bottom layer links. This
contradicts that any two links on the bottom layer should carry four
independent information to achieve rate 4. Therefore forwarding is
insufficient for achieving the rate 4 in this network.

Now we show that a generic linear network code, where intermediate
nodes do coding, achieves rate 4. From \cite[Ch
19]{yeung2008information}, a  generic network code can be constructed
with high probability by randomly choosing the global encoding
kernels provided that the base field is sufficiently large.
So a random linear network code is
generic with high probability when $q$ is very large. If the adversarial link is on the top or middle layer, then each capacity 2 on the
bottom layer is equivalent to two unit capacity links. Then all
links in the network have capacity one and this problem is reduced
to the equal link capacities problem. From \cite{cai2006network},
rate $6-2\times 1=4$ is achievable. From ~\cite[Theorem
19.32]{yeung2008information}, since the min-cut between $s$ and
$(Y_i,Y_j)$ is at least 4 for $\forall 1\leq i\neq j\leq 4$, in a
generic network code the global encoding kernels on any two links on
the bottom layer are linearly independent and they satisfy the MDS
property. Thus an error on the last layer can be corrected.
\end{proof}

\subsection{Guess-and-forward}\label{section:guess}
Here we introduce a new achievable coding strategy,
guess-and-forward. In this strategy, a node receives some redundant
information from multiple paths.  If this information is
inconsistent, the node guesses which of its upstream links
controlled by the adversary and forwards its guess to the sink.  The
sink receives additional information allowing it to test the
hypothesis of the guessing node and correctly identify one or more
adversarial links. Altogether a finite number of guesses are
forwarded, so the average overhead of forwarding guesses goes to
zero asymptotically with the total amount of information
communicated.

To provide intuition, we first describe a simple version of the
guess-and-forward scheme on a particular four-node network example
shown in Fig.~\ref{fig610}.  From Theorem \ref{bound}, when $z=2$,
the capacity is upper bounded by 7. We will show that rate $7$ is
achievable in this network using the guess-and-forward scheme.

In this scheme, node $A$  forwards its received information to node
$B$ and on multiple links to the sink node $t$. This information is
received reliably by the sink, but not necessarily by $B$ since the
single feedback link from $A$ to $B$ may be adversarial. Node $B$
also receives reliably from the source node $s$ the information that
was sent from $s$ to $A$, and compares this with the information
forwarded by $A$. A mismatch indicates that either the feedback link
$(A,B)$ is adversarial or that one or more links from $s$ to $A$ are
adversarial. $B$ sends this  guess to the sink along with the
information received on link $(A,B)$, which allows the sink to
distinguish between  the two possibilities.
Note that decoding at the sink relies on knowledge of the network and code, while the guessing node simply has to compare the feedback with the corresponding information forwarded from the source node.

\begin{figure}
\begin{center}
\begin{picture}(200,200)(-40,-20)\centering
\put(75,150){\circle{7}}\put(75,15){\circle{7}}\put(10,82){\circle{7}}\put(140,82){\circle{7}}
\qbezier(75,150)(50,100)(10,82)
\qbezier(75,150)(20,160)(10,82)\qbezier(75,150)(30,130)(10,82)
\qbezier(75,150)(100,100)(140,82)\qbezier(75,150)(80,100)(140,82)
\qbezier(75,150)(130,160)(140,82)\qbezier(75,150)(120,130)(140,82)\put(75,15){\line(65,67){65}}
\qbezier(10,82)(50,65)(75,15)\qbezier(10,82)(70,65)(75,15)
\qbezier(10,82)(20,5)(75,15)\qbezier(10,82)(30,35)(75,15)
\qbezier(140,82)(100,65)(75,15)\qbezier(140,82)(80,65)(75,15)
\qbezier(140,82)(130,5)(75,15)\qbezier(140,82)(120,35)(75,15)
\put(10,82){\line(1,0){130}}
\put(75,160){$s$}\put(75,3){$t$}\put(0,90){$A$} \put(140,90){$B$}
\put(58,149){\vector(-2,-1){0}}\put(58,141){\vector(-2,-1){0}}
\put(65,133){\vector(-1,-1){0}}
\put(97,149){\vector(2,-1){0}}\put(85,133){\vector(1,-2){0}}
\put(91,141){\vector(1,-1){0}}\put(80,133){\vector(1,-3){0}}
\put(32,74){\vector(2,-1){0}}\put(13,66){\vector(1,-2){0}}
\put(29,72){\vector(2,-1){0}}\put(18,67){\vector(1,-2){0}}
\put(50,82){\vector(1,0){0}}
\put(110,70){\vector(-2,-1){0}}\put(113,66){\vector(-1,-1){0}}\put(120,62){\vector(-1,-1){0}}
\put(126,58){\vector(-1,-1){0}}\put(132,54){\vector(-1,-1){0}}
\put(33,128){2} \put(41,120){2}\put(51,110){2}\put(70,90){6}
\put(107,48){1}\put(113,42){1}\put(123,32){1}\put(101,54){1}\put(95,60){1}
\put(10,32){$\infty$}\put(130,132){$\infty$} \put(24,137){$l_1$}
\put(37,111){$l_2$}\put(47,101){$l_3$}
\end{picture}
\end{center}
\caption{Four node acyclic networks: this network consists of $3$
links of capacity 2 from $s$ to $A$, $5$ links of capacity 1 from
$B$ to $t$, $1$ links of capacity 6 from $A$ to $B$. Given the cut
$(\{s,B\},\{A,t\})$, unbounded reliable communication is allowed
from source $s$ to its neighbor $B$ on one side of the cut and from
node $A$ to sink $t$ on the other side of the cut, respectively.
}\label{fig610}
\end{figure}
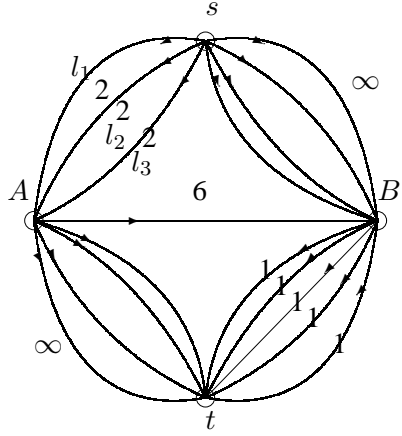

To be precise, let $W$ and $\hat{W}$ denote the symbols sent by $s$
and received by $A$ respectively on the links from $s$ to $A$. $W$
is sent reliably from $s$ to $B$, and $\hat{W}$ is sent reliably
from $A$ to $t$.

In each round, $s$ and $B$ together send a $(11,7)$ MDS code to $A$
and $t$ across the cut $cut(\{s,B\},\{A,t\})$. Since the feedback
link has capacity 6, $A$ sends its codeword symbols $\hat{W}$ to $B$
along feedback link $l$. For feedback link $l$, let $P_l$ denote the
information received by $B$ on $l$. $B$ compares $P_l$ with $W$
which is received from $s$. If $P_l\neq W$, then $B$ obtains a guess
$X_l$ identifying the locations of adversarial links between $s$ and
$A$ assuming $P_l$ is reliable. $B$ sends the claim $(X_l,P_l)$ to
$t$ along each link between $B$ and $t$ using repetition code. If
$P_l=W$, $B$ does not send any claim. $B$ sends claims only when it
guesses at least one adversarial forward link which is different
from those guessed in previous rounds. Thus $B$ sends claims in 3
rounds, which is equal to the number of links between $s$ and $A$.
Note that on a result the number of channel uses is not constant in
each round. Since the total number of channel uses required to send
claims is finite, the overhead amortized over a large number of
rounds goes to zero asymptotically with the number of rounds.

\begin{lemma}\label{guessforwardex}
Given the network in Fig.~\ref{fig610}, rate 7 is achievable.
\end{lemma}
\begin{proof}
Since there are $5\geq 2z+1$ links from $B$ to $t$, any claim
$(X_l,P_l)$ can be sent reliably from $B$ to $t$ using a repetition
code.

Case 1) the sink receives some claim $(X_l,P_l)$.

The sink compares $P_l$ with $\hat{W}$ which is received from $A$
reliably. If $P_l\neq \hat{W}$, then the feedback link transmitting
$P_l$ is adversarial and the sink ignores it. Otherwise, $P_l$ is
reliable. Since the claim is sent, the sink knows that
$P_l=\hat{W}\neq W$ and that guess $X_l$ is correct. Thus the sink
identifies the forward links in $X_l$ as adversarial, which are
subsequently ignored.

Case 2) no claims are sent.

In this case, we show that the correct output is achieved. No claims are sent
only if either\begin{itemize}\item $B$ receives  $W$ on the feedback link, or\item the guessed set $X_l$  only contains
forward links that have been guessed by $l$ in previous rounds.
From these previous rounds, by case 1, the sink
has already identified as adversarial either $l$ or the guessed
forward links, and is concerned only with the remaining network.
\end{itemize}
Either way,   there
are the following two possibilities for the overall remaining network (recall that $A$ transmits $\hat{W}$ to $B$):

(I) all links between $s$ and $A$ and the feedback link are
uncorrupted.

(II) some links between $s$ and $A$ are corrupted and feedback link
is corrupted such that feedback link transmits error-free output.

In possibility (I), the feedback link transmits $W$ to $B$. In (II),
$A$ sends $\hat{W}\neq W$ but the feedback link changes it to $W$ so
that $B$ does not send any claims. We first consider all sets of 7
forward links on the cut. There are ${{8}\choose {7}} = 8$ such sets
of links. Each set has total capacity at least 9. For each such set
$L$, the sink checks the consistency of the output of rate 7
obtained from $L$. We also consider all sets of $6$ links such that
each set includes all $3$ links between $s$ and $A$ and any $3$
links between $B$ and $t$. There are ${5}\choose {3}$ such sets. The
sink also checks the consistency of the output of rate $7$ for each
set.

Case 2 - a) there is no set of $7$ links giving consistent output.

In this case, there are more than $1$ forward adversarial link on
the cut. Since $z=2$, all two adversarial links are forward links
and thus possibility (II) cannot hold.  Then possibility (I) is true
and there are at most two forward adversarial links with capacity 1
on the cut. We obtain the correct answer from our (11,7) MDS code.

Case 2 - b) there is no set of $6$ links that includes all $3$ links
from $s$ to $A$ and gives consistent output.

In this case, possibility (II) is true. Then there is at most one
forward adversarial link on the cut. We obtain the correct answer
from our (11,7) MDS code.

Case 2 - c) There exist both a $7$-link set $L_1$ giving consistent
output and a $6$-link set $L_2$ that includes all $3$ links between
$s$ and $A$ and give consistent output.

It is clear that $\sum_{l_1\in L_1\cap L_2}r(l_1)\geq 7$ for any
$L_1$ and $L_2$. Thus $L_1$ and $L_2$ give the same consistent
output. Since at least one of (I) and (II) is true, this output is
correct.

From cases 1-2, since $z=2$, $B$ needs to send claims at most $2$
times to obtain the correct output.
\end{proof}

\section{Guess-and-forward on some families of networks}\label{section:networks}
In this section, we employ the guess-and-forward strategy on a
sequence of increasingly complex network families. The first is a two-node
network with multiple feedback links. The second is a four-node
acyclic network. The third is a family of `zig-zag' networks. In the
first two cases, the guess-and-forward strategy achieves the
capacity. For zig-zag networks, we derive the achievable rate of
guess-and-forward strategy and present conditions under which this
bound is tight.

\subsection{Two-node network} We achieve the
error-correction capacity of the two-node network with multiple
feedback links by using guess-and-forward strategy. A two-node
network shown in Fig.~\ref{fig611} is composed of $n$ forward links
with arbitrary capacity and $m$ feedback links with arbitrary
capacity. In Lemma~\ref{feed}, we first characterize the capacity of
this network when each forward link has capacity 1. We extend this
result to Theorem~\ref{feeduneq} when each forward link has
arbitrary capacity.






\begin{lemma}\label{feed}
Consider the two-node network shown in Fig.~\ref{fig611} such that
each forward link has capacity 1. Let $C$ denote the
error-correction capacity with $z$ adversarial links. If $n\leq 2z$,
$C=0$. Otherwise, $C=\min\{n-z, n-2(z-m)\}$.
\end{lemma}

\begin{proof}

From lemma~\ref{twonodeup} in Section~\ref{section:upper}, upper
bound of the capacity is $\min\{n-z, n-2(z-m)\}$ when $n>2z$, and 0
otherwise. So it is sufficient to prove the achievability of this
upper bound by applying our guess-and-forward strategy when $n>2z$.

\begin{figure}
\begin{center}
\begin{picture}(200,200)(-40,-20)\centering
\put(75,150){\circle{7}}\put(75,15){\circle{5}}\qbezier(75,150)(-20,82)(75,15)
\qbezier(75,150)(20,82)(75,15)\qbezier(75,150)(130,82)(75,15)\qbezier(75,150)(170,82)(75,15)
\put(30,82){\circle{1}}\put(36,82){\circle{1}}\put(42,82){\circle{1}}
\put(108,82){\circle{1}}\put(114,82){\circle{1}}\put(120,82){\circle{1}}
\put(36,72){$n$}\put(110,72){$m$}
\put(75,160){$s$}\put(75,3){$t$}
\put(28,80){\vector(0,-1){0}}\put(48,80){\vector(0,-1){0}}
\put(103,80){\vector(0,1){0}}\put(123,80){\vector(0,1){0}}
\end{picture}
\end{center}
\caption{ two-node network $\mathcal{G}$ with $n$ forward links and
$m$ feedback links.}\label{fig611}
\end{figure}
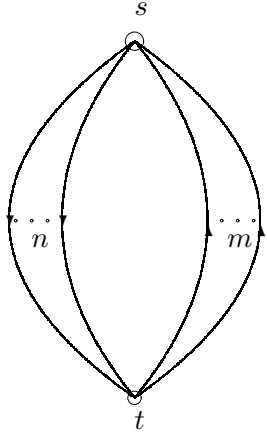

Case 1) $m\leq \frac{z}{2}$.

Step 1) In each round, the source $s$ sends an $(n,n-2(z-m))$
MDS code on the $n$ forward links. Since $m\leq z/2$, $n-2(z-m)\leq
n-z$. Thus for any received $n$ signals, there exist $n-2(z-m)$
uncorrupted signals. If all $n \choose {n-2(z-m)}$ subsets of
received symbols decode to the same message, this message is
correct. Otherwise, the sink sends the $n$ received signals to the
source $s$ on each feedback link using a repetition code.

Step 2) Based on the received information on each feedback link, the
source tries to identify the bad forward links. Thus, for each
feedback link, the source obtains a claim regarding the location of
forward adversarial links which is correct if that feedback link is
not adversarial.

Step 3) This step consists of $m$ phases, each composed of a finite
number of rounds. In the $i$th phase, the source sends the claim
obtained from the $i$th feedback link together with what it received
on that feedback link to the sink. This information can be sent
reliably to the sink using a repetition code because $n-2z>0$. If
what the source received matches what the sink sent, the $i$th
feedback link was not corrupted and the associated claim is correct.
Using this claim, the sink can decode the message as well as
identify at least one of the forward adversarial links. If all $m$
feedback links were corrupted, the sink knows that there are only
$z-m$ forward adversarial links and since we are using a
$(n,n-2(z-m))$ MDS code the message is correctly decodable at the
sink.

Note that we only need to use the above scheme during the first $2m$
times the sink sees inconsistency at step 1. The reason is that from
steps 1-3, the sink either figures out that all feedback links are
adversarial or identifies at least one forward adversarial link. If
all feedback links are bad, they are ignored and the $(n,n-2(z-m))$
MDS code gives us the correct output. If there are $k\leq 2m$
forward adversarial links, after the first $k$ times the sink sees
inconsistency at step 1, all forward adversarial links are
identified subsequently and no further inconsistency is seen among
the remaining forward links. Otherwise, when there are more than
$2m$ adversarial links, the sink finds $2m$ forward adversarial
links and ignores them. Then from \cite{cai2006network}, the rate
$n-2m-2(z-2m)=n-2(z-m)$ can be achieved using the remaining forward
links only.

Case 2) $m> \frac{z}{2}$.

In each round, the source $s$ sends an $(n,n-z)$ MDS code on the
$n$ forward links. For any received $n$ signals, there exist $n-z$
uncorrupted signals. If all $n \choose {n-z}$ subsets of received
symbols decode to the same message, this message is correct. As in
the case 2-a, from steps 2-3, the sink either concludes that all
feedback links are adversarial or identifies at least one forward
adversarial link. If all $m$ feedback links were corrupted, there
are only $z-m<z/2$ bad forward links and subsequently only the
forward links are used to achieve the rate $n-z$. Otherwise, the
above scheme is used at most $z$ times inconsistency is seen at step
1, after which the sink has identified all bad forward links and the
remaining forward links suffice to achieve rate $n-z$.
\end{proof}

Now we generalize above result to the general case when each forward
link has also arbitrary capacity.

\begin{theorem}\label{feeduneq}
Consider the two-node network shown in Fig.~\ref{fig611} with
arbitrary link capacities. Let $D_p$ denote the sum of the $p$
smallest forward link capacities. The error-correction capacity is
\[ C =\left\{\begin{array}{ll}
            0 & \mbox{if $n\leq 2z$}\\
           {\min\{D_{n-z},D_{n-2(z-m)^+}\}} &\mbox{if $n>2z$}\\
            \end{array}
             \right.\]
\end{theorem}

\begin{proof}
From lemma~\ref{twonodeup}, achievable capacity is 0 when $n\leq
2z$. When $n>2z$, we use $D$ to denote the sum of all $n$ forward
link capacities. For achievability, when $m\leq z/2$, the source
sends $(D,D_{n-z})$ MDS code to the sink. When $m>z/2$, the source
sends $(D,D_{n-2(z-m)^+})$ MDS code to the sink. By using the same
strategy as in the proof of Lemma \ref{feed}, we can achieve the
rate $C$.

\end{proof}

\subsection{Four-node acyclic network}
\begin{figure}
\begin{center}
\begin{picture}(200,200)(-40,-20)\centering
\put(75,150){\circle{7}}\put(75,15){\circle{7}}\put(10,82){\circle{7}}\put(140,82){\circle{7}}
\qbezier(75,150)(50,100)(10,82)\qbezier(75,150)(70,100)(10,82)
\qbezier(75,150)(20,160)(10,82)\qbezier(75,150)(30,130)(10,82)
\qbezier(75,150)(100,100)(140,82)\qbezier(75,150)(80,100)(140,82)
\qbezier(75,150)(130,160)(140,82)\qbezier(75,150)(120,130)(140,82)
\qbezier(10,82)(50,65)(75,15)\qbezier(10,82)(70,65)(75,15)
\qbezier(10,82)(20,5)(75,15)\qbezier(10,82)(30,35)(75,15)
\qbezier(140,82)(100,65)(75,15)\qbezier(140,82)(80,65)(75,15)
\qbezier(140,82)(130,5)(75,15)\qbezier(140,82)(120,35)(75,15)
\qbezier(10,82)(75,100)(140,82)\qbezier(10,82)(80,65)(140,82)
\put(75,89){\circle{1}}\put(75,82){\circle{1}}\put(75,75){\circle{1}}
\put(75,160){$s$}\put(75,3){$t$}\put(0,90){$A$} \put(140,90){$B$}
\put(58,149){\vector(-2,-1){0}}\put(58,141){\vector(-2,-1){0}}
\put(65,133){\vector(-1,-1){0}}\put(71,133){\vector(-1,-2){0}}
\put(97,149){\vector(2,-1){0}}\put(85,133){\vector(1,-2){0}}
\put(91,141){\vector(1,-1){0}}\put(80,133){\vector(1,-3){0}}
\put(32,74){\vector(2,-1){0}}\put(13,66){\vector(1,-2){0}}
\put(29,72){\vector(2,-1){0}}\put(18,67){\vector(1,-2){0}}
\put(50,90){\vector(2,1){0}}\put(50,74){\vector(2,-1){0}}
\put(110,70){\vector(-2,-1){0}}\put(113,66){\vector(-1,-1){0}}
\put(126,58){\vector(-1,-1){0}}\put(132,54){\vector(-1,-1){0}}
\put(10,32){$\infty$}\put(130,132){$\infty$}
\put(10,150){\line(26,-27){130}}\put(145,10){$Q$}
\end{picture}
\end{center}
\caption{Four node acyclic networks: unbounded reliable
communication is allowed from source $s$ to its neighbor $B$  and
from node $A$ to sink $t$, respectively. This network consists of
$a$ links of arbitrary capacity from $s$ to $A$, $b$ links of
arbitrary capacity from $B$ to $t$. From $A$ to $B$, there are $m$
feedback links and each feedback link has the minimum
capacity.}\label{fig612}
\end{figure}
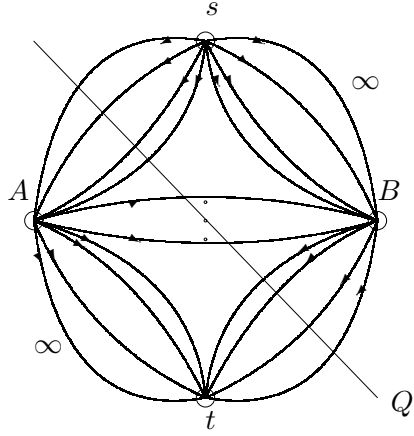

In this section, we use the guess-and-forward strategy on
a four-node acyclic network, the simplest case of a zigzag network. In this acyclic
network, source node $s$ and its neighbor node $B$ lie on one side
of a cut that separates them from sink node $t$ an its neighbor $A$.
As in the cut-set model, we allow unbounded reliable communication
from source $s$ to its neighbor $B$ on one side of the cut and from
node $A$ to the sink $t$ on the other side of the cut; this allows
node $B$ to compare information from feedback links with
uncorrupted information from the source to obtain the guess.
Similarly, by comparing claims with information reliably transmitted
from node $A$, the sink can identify the corrupted link.

This network is composed of a set of $a$ forward links
$\{l_1,..,l_a\}$ with arbitrary capacities from $s$ to $A$, a set of
$b$ forward links $\{l_{a+1},..,l_{a+b}\}$ with arbitrary capacities
from $B$ to sink $t$, and a set of $m$ feedback links from $A$ to $B$.  Each
feedback link has capacity $h$ whose value will be derived
in following Section~\ref{subsec:fournodefeedcap}.
$C_1=\sum_{l\in(l_1,..,l_{a})}r(l)$ and
$C_2=\sum_{l\in(l_{a+1},..,l_{a+b})}r(l)$ denotes the sum of forward
link capacities from $s$ to $A$, and from $B$ to $t$, respectively.
Let $C=C_1+C_2$. $C_z$ is the upper bound on this network obtained
from Theorem~\ref{bound}.

In \cite{kim2010nec}, we have shown that rate $C_z$ is
asymptotically achievable on this four-node acyclic network when
each feedback link has capacity at least $C_1$. Here we show that
rate $C_z$ is achievable even when each feedback link has smaller
capacity than $C_1$. In Section~\ref{subsec:fournodefeedcap}, we
first consider a coding strategy at node $A$ and formulate a
linear optimization problem which gives the minimum capacity of each
feedback link that guarantees the success of our strategy. Then, in
Section~\ref{subsec:fournoderesults}, we show that rate $C_z$ is
asymptotically achievable this smaller feedback link capacity using
our guess-and-forward strategy.

\begin{figure*}
\begin{center}
\begin{picture}(200,200)(-40,-20)\centering
\put(-30,150){\circle{7}}\put(-30,15){\circle{7}}\put(-95,82){\circle{7}}\put(35,82){\circle{7}}
\qbezier(-30,150)(-35,100)(-95,82)
\qbezier(-30,150)(-85,140)(-95,82)
\qbezier(-30,150)(-5,100)(35,82)\qbezier(-30,150)(-25,100)(35,82)
\qbezier(-30,150)(25,160)(35,82)\qbezier(-30,150)(15,130)(35,82)
\qbezier(-95,82)(-55,65)(-30,15)\qbezier(-95,82)(-35,65)(-30,15)
\qbezier(-95,82)(-85,5)(-30,15)\qbezier(-95,82)(-75,35)(-30,15)
\qbezier(35,82)(-5,65)(-30,15)\qbezier(35,82)(-25,65)(-30,15)
\qbezier(35,82)(25,5)(-30,15)\qbezier(35,82)(15,35)(-30,15)
\put(-95,82){\line(65,68){65}}\put(-95,82){\line(1,0){130}}
\put(-70,140){2}\put(-42,115){2}\put(-56,128){2}
\put(-90,120){$a_1,a_2$}\put(-62,95){$c_1,c_2$}\put(-76,108){$b_1,b_2$}
\put(-100,110){$l_1$}\put(-72,88){$l_3$}\put(-86,98){$l_2$}
\put(-80,70){$a_1+a_2,b_1+b_2,c_1+c_2$}\put(-35,85){3}
\put(-20,60){1}\put(-10,50){1}\put(0,40){1}\put(10,30){1}
\put(-30,160){$s$}\put(-30,3){$t$}\put(-105,90){$A$}
\put(35,90){$B$}
\put(180,150){\circle{7}}\put(180,15){\circle{7}}\put(115,82){\circle{7}}\put(245,82){\circle{7}}
\qbezier(180,150)(155,100)(115,82)\qbezier(180,150)(175,100)(115,82)
\qbezier(180,150)(125,160)(115,82)\qbezier(180,150)(135,130)(115,82)
\qbezier(180,150)(205,100)(245,82)\qbezier(180,150)(185,100)(245,82)
\qbezier(180,150)(235,160)(245,82)\qbezier(180,150)(225,130)(245,82)
\qbezier(115,82)(155,65)(180,15)\qbezier(115,82)(175,65)(180,15)
\qbezier(115,82)(125,5)(180,15)\qbezier(115,82)(135,35)(180,15)
\qbezier(245,82)(205,65)(180,15)\qbezier(245,82)(190,75)(180,15)
\qbezier(245,82)(235,15)(180,15)\qbezier(245,82)(225,35)(180,15)\qbezier(245,82)(235,-5)(180,15)
\put(180,150){\line(-65,-68){65}}\put(180,15){\line(65,67){65}}
\put(140,140){6}\put(150,132){6}\put(156,125){4}\put(162,120){4}\put(168,115){3}
\put(125,125){$l_1$}\put(135,117){$l_2$}\put(141,110){$l_3$}\put(147,105){$l_4$}\put(153,100){$l_5$}
\put(190,60){1}\put(196,54){1}\put(202,48){1}\put(210,40){1}\put(220,30){1}\put(228,22){1}
\put(115,82){\line(1,0){130}}\put(180,85){5}
\put(180,160){$s$}\put(180,3){$t$}\put(105,90){$A$}
\put(245,90){$B$}
\put(-47,144){\vector(-2,-1){0}}\put(-34,133){\vector(-1,-2){0}}\put(-42,137){\vector(-1,-1){0}}
\put(-13,141){\vector(2,-1){0}}\put(-25,133){\vector(1,-2){0}}
\put(-6,149){\vector(2,-1){0}}\put(-22,135){\vector(1,-2){0}}
\put(-77,76){\vector(2,-1){0}}\put(-90,66){\vector(1,-2){0}}
\put(-79,74){\vector(2,-1){0}}\put(-87,67){\vector(1,-2){0}}
\put(-60,82){\vector(1,0){0}}
\put(-2,64){\vector(-2,-1){0}}\put(2,60){\vector(-2,-1){0}}
\put(18,52){\vector(-1,-1){0}}\put(27,52){\vector(-1,-1){0}}
\put(158,149){\vector(-2,-1){0}}\put(176,133){\vector(-1,-2){0}}
\put(162,140){\vector(-2,-1){0}}\put(165,135){\vector(-1,-1){0}}\put(171,135){\vector(-1,-1){0}}
\put(204,149){\vector(2,-1){0}}\put(188,135){\vector(1,-2){0}}
\put(207,65){\vector(-1,-1){0}}\put(212,59){\vector(-1,-1){0}}\put(218,54){\vector(-1,-1){0}}
\put(225,48){\vector(-1,-1){0}}\put(231,42){\vector(-1,-1){0}}\put(234,39){\vector(-1,-1){0}}
\put(197,141){\vector(2,-1){0}}\put(185,133){\vector(1,-2){0}}
\put(133,76){\vector(2,-1){0}}\put(123,66){\vector(1,-2){0}}
\put(131,74){\vector(2,-1){0}}\put(119,66){\vector(1,-3){0}}
\put(150,82){\vector(1,0){0}}
\put(20,135){$\infty$}\put(-100,40){$\infty$}
\put(230,135){$\infty$}\put(110,40){$\infty$}
\end{picture}
\end{center}
\caption{Four node acyclic networks: (a) $z=2$ and feedback link
transmits $(a_1+a_2,b_1+b_2,c_1+c_2)$. (b) $z=3$. Assume that
$(a_1,..,a_6)$, $(b_1,..,b_6)$, $(c_1,..,c_4)$, $(d_1,..,d_4)$, and
$(e_1,e_2,e_3)$ are transmitted on forward links $(l_1,..,l_5)$ from
$s$ to $A$, respectively. Feedback link transmits $(\sum_{i=1}^6
a_i,\sum_{i=1}^6 b_i,\sum_{i=1}^4 c_i,\sum_{i=1}^4 d_i,\sum_{i=1}^3
e_i)$.}\label{fig613}
\end{figure*}
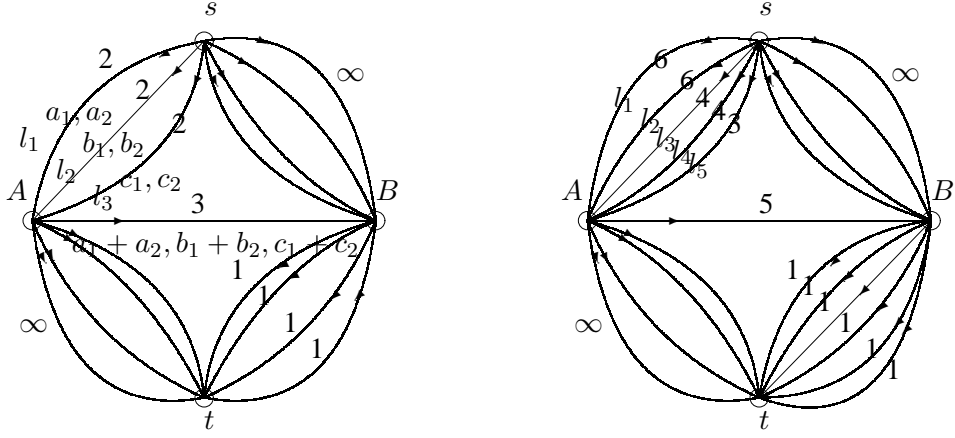


\subsubsection{Coding strategy at node $A$}\label{subsec:fournodefeedcap}

Suppose that $(s,B)$ sends $(C,C_z)$ MDS code across the cut to
$(A,t)$. We consider the encoding strategy at node $A$ and derive
the minimum capacity of each feedback link. Suppose that node $A$
receives the vector of symbols
$\hat{W}=(\hat{W}_{l_1},\ldots,\hat{W}_{l_a})$ from $s$ where
$\hat{W}_l=(p_l^1,\ldots,p_l^{r(l)})$ denotes the codewords on link
$l\in\{l_1,\ldots,l_a\}$. We first assume that node $A$ transmits,
on each feedback link to $B$, the same set of codewords each of
which is a linear combination of codewords received on a single link
from $s$ to $A$. Precisely, for any forward link $l_i$, node $A$
transmits on each feedback link
$g(\hat{W}_{l_i})=(g_{l_i}^1(\hat{W}_{l_i}),\ldots,g_{l_i}^{k_i}(\hat{W}_{l_i}))$
where $g_{l_i}^j(\hat{W}_{l_i})$ is a single linear combination of
$\hat{W}_{l_i}=(p_{l_i}^1,\ldots,p_{l_i}^{r(l_i)})$. Thus, the same
value $g(\hat{W})=(g(\hat{W}_{l_1}),..,g(\hat{W}_{l_a}))$ is
transmitted on each feedback link. For instance, given a network in
Fig.~\ref{fig613}(a), $A$ transmits
$g(\hat{W})=(g(\hat{W}_{l_1}),g(\hat{W}_{l_2}),g(\hat{W}_{l_3}))$
where $g(\hat{W}_{l_1})=a_1+a_2$, $g(\hat{W}_{l_2})=b_1+b_2$, and
$g(\hat{W}_{l_3})=c_1+c_2$.


Here, we define the degree of freedom of forward link $l$ between
$s$ and $A$ as follows.

\begin{definition}\label{degfree}
Consider the vector of symbols $\hat{W}_l$ received on forward link
$l$ from $s$ to $A$ and assume that node $A$ transmits $k$ linear
combinations of $\hat{W}_l$,
$g(\hat{W}_l)=(g_l^1(\hat{W}_l),\ldots,g_l^k(\hat{W}_l))$. Let $M_l$
denote the $r(l)\times k$ encoding matrix at $A$ for link $l$ such
that $\hat{W}_l\cdot M_l=g(\hat{W}_l)$. Then the degree of freedom
of link $l$, $f(l)$, is defined as the capacity of link $l$ minus
the rank of the matrix $M_l$, i.e., $f(l) = r(l)-rank(M_l)$. For any
forward link $l$ between $B$ and $t$, we simply define the degree of
freedom $f(l)$ as the link capacity, i.e., $f(l)=r(l)$.
\end{definition}

For example, in Fig.~\ref{fig613}(a), since feedback link transmits
$(a_1+a_2,b_1+b_2,c_1+c_2)$, $f(l)=1$ for all forward links from $s$
to $A$. In Fig.~\ref{fig613}(b), since feedback link transmits
$(\sum_{i=1}^6 a_i,\sum_{i=1}^6 b_i,\sum_{i=1}^4 c_i,\sum_{i=1}^4
d_i,\sum_{i=1}^3 e_i)$, $f(l_1)=f(l_2)=5$, $f(l_3)=f(l_4)=3$, and
$f(l_5)=2$.

From the definition of degree of freedom, node $A$ sends
\begin{equation}\label{fournodeeq1}
h = \sum_{l\in\{l_1,..,l_a\}} (r(l)-f(l)) =
C_1-\sum_{l\in\{l_1,..,l_a\}} f(l)
\end{equation}
codewords to $B$ along each feedback link.

Now we introduce our coding strategy at node $A$ as follows.

Node $A$ can choose any $g(\hat{W})$ which satisfies the 
following two conditions on the degree of freedom of links.

\begin{condition}\label{cond1}
Given any set $A_1$ composed of $2z$ forward links, $\sum_{l\in
A_1}f(l)\leq C-C_z$.
\end{condition}

\begin{condition}\label{cond2}
Given any set $A_2$ composed of $z$ forward links and $A_3$ composed
of $z-m$ forward links such that $A_2\cap A_3=\emptyset$,
$\sum_{l\in A_2}f(l)+\sum_{l\in A_3}r(l)\leq C-C_z$.
\end{condition}
Condition~\ref{cond1} means that the sum of the degree of freedom of
any $2z$ forward links are less than or equal to $C-C_z$.
Condition~\ref{cond2} means that the sum of the degree of freedom of
any $z$ links plus the sum of any $z-m$ link capacities is less than
or equal to $C-C_z$.
In the proof of Lemma~\ref{suffbu} and~\ref{insuffbu}, we show that
these two conditions are necessary to prove the tightness of our
upper bound in Theorem~\ref{bound}. For example network in
Fig.~\ref{fig613}(a), $z=2$ and the upper bound $C_z=6$. 3 codewords
sent by $A$ satisfies above two conditions, and feedback capacity 3
is sufficient. Likewise, when $z=3$ and the upper bound $C_z=9$ in
the network Fig.~\ref{fig613}(b), 5 codewords sent by $A$ also
satisfies above two conditions. In~\cite{kim2010nec}, the minimum
required capacity for each feedback link to achieve rate $C_z$ is
the sum of all forward link capacities between $s$ and $A$, which is
6 and 23 for the networks in Fig.~\ref{fig613}(a) and (b),
respectively.

Finally, we formulate a linear optimization problem which gives the
minimum capacity of each feedback link, based on
conditions~\ref{cond1} and~\ref{cond2}.


\begin{equation}\label{eq2}
\begin{split}
\text{min}\quad& h = C_1-\sum_{i=1}^a f(l_i)\\
&f(l_i)\leq r(l_i), \quad \forall 1\leq i\leq a+b\\
&\sum_{l\in M}f(l)\leq C-C_z, \quad M\subset\mathcal{E}, |M|\leq 2z \\
&\sum_{l\in N_1}r(l)+\sum_{l\in N_2}f(l)\leq C-C_z,\\
&\quad N_1,N_2\subset\mathcal{E},|N_1|\leq z-m, |N_2|\leq z, N_1\cap N_2=\emptyset\\
\end{split}
\end{equation}

Objective function $h$ is defined in equation~(\ref{fournodeeq1}).
The first inequality constraint is the link capacity constraint. The second
and third constraints come from condition~\ref{cond1}
and~\ref{cond2}, respectively. We can check that solving the above
optimization problem for the networks in Fig.~\ref{fig613}(a) and
(b) gives the feedback link capacities 3 and 5, respectively.

\subsubsection{Guess-and-forward strategy}\label{subsec:fournoderesults}
In this section, we show the tightness of the upper bound using our
guess-and-forward strategy. Our proof of decoding success requires
each feedback link's capacity to satisfy the lower bound obtained
in~(\ref{eq2}), so that node $B$ receives sufficient feedback
information to guess the corrupted links by comparing with
information reliably received from the source. This feedback link
capacity can in general be smaller than that in the simple example
of Section~\ref{subsec:fournodefeedcap}, so the details and proof of
correctness of the scheme are slightly more involved.

The guess-and-forward strategy achieving the rate $C_z$ for the four-node
acyclic network shown in Fig.~\ref{fig612} is as follows.
In each round, $s$ and $B$ together send a $(C,C_z)$ MDS code\footnote{A generic linear code is MDS.},
obtained from Lemma~\ref{fdbackmds} below, to $A$ and $t$ across the cut
$cut(\{s,B\},\{A,t\})$. Let $W$ and $\hat{W}$ denote the codewords
$s$ sends to $A$, and $A$ received from $s$, respectively. Using the
coding strategy in Section~\ref{subsec:fournodefeedcap}, $A$ sends
$g(\hat{W})$ to $B$ along each feedback link using a repetition
code. For each feedback link $l$, let $P_l$ denote the information
received by $B$ on $l$. $B$ compares $P_l$ with $g(W)$. If $P_l\neq
g(W)$, then $B$ obtains a guess $X_l$ identifying the locations of
adversarial links between $s$ and $A$ assuming $P_l$ is reliable. $B$ does not send  any claim if
$P_l=g(W)$, or if the guessed set $X_l$  only contains
forward links that have been guessed by $l$ in previous rounds.

We next show that this strategy achieves rate $C_z$ asymptotically, via the following series of lemmas.

\begin{lemma}\label{fdbackmds}
Given the four-node acyclic network in Fig.~\ref{fig612}, let $u$
denote the sum of $2z$ largest degree of freedom of links in the
network. Suppose the adversary introduces errors on $z$ forward links
subject to the constraint that the  values sent along the feedback links are unaffected.  There exists a 
$(C,C-u)$ generic linear code for the forward links that corrects these $z$ error links.
\end{lemma}
\begin{proof}
See the appendix.
\end{proof}

Since the sum of $2z$ largest degree of freedom is at most $C-C_z$
from the condition~\ref{cond1}, we obtain $C-u\geq C_z$.

\begin{definition}\label{defconsistency}
Given a set of any $k$ forward links $L=\{l_1,\ldots,l_k\}$ in the
four-node acyclic network, we say $L$ gives consistent output if
$\sum_{l\in L}r(l)\geq C_z$ and the decoded output from any $C_z$
code symbols on $L$ is the same.
\end{definition}

\begin{lemma}\label{suffbu}
Given the four-node network in Fig.~\ref{fig612} such that $b\geq
2z+1$, rate $C_z$ is achievable.
\end{lemma}
\begin{proof}
Since $b\geq 2z+1$, any claim $(X_l,P_l)$ can be sent reliably from
$B$ to $t$ using a repetition code. The details of proof is
presented in the appendix.
\end{proof}

\begin{lemma}\label{insuffbu}
Given the four-node network in Fig.~\ref{fig612} such that $b\leq
2z$, rate $C_z$ is achievable.
\end{lemma}

\begin{proof}
When $b\leq 2z$, reliable transmission of claims from $B$ to $t$ is
not guaranteed. Thus we cannot use the same technique used in the
proof of Lemma~\ref{suffbu}. The proof is presented in the appendix.
\end{proof}

\subsection{zig-zag network}

\begin{figure}
\begin{center}
\begin{picture}(200,200)(-40,-20)\centering
\put(-55,40){\circle{5}}\put(-15,40){\circle{5}}\put(25,40){\circle{5}}\put(135,40){\circle{5}}\put(175,40){\circle{5}}
\put(-55,110){\circle{5}}\put(-15,110){\circle{5}}\put(25,110){\circle{5}}\put(135,110){\circle{5}}\put(175,110){\circle{5}}
\qbezier(-55,110)(-70,90)(-55,40)\qbezier(-55,110)(-40,90)(-55,40)\put(-55,40){\line(0,1){70}}
\qbezier(-55,110)(-35,100)(-15,110)\qbezier(-55,110)(-35,120)(-15,110)
\qbezier(-55,40)(-40,60)(-15,110)\qbezier(-55,40)(-40,90)(-15,110)
\qbezier(-55,40)(-35,30)(-15,40)\qbezier(-55,40)(-35,50)(-15,40)
\qbezier(-15,110)(-30,90)(-15,40)\qbezier(-15,110)(0,90)(-15,40)\put(-15,40){\line(0,1){70}}
\qbezier(-15,110)(5,100)(25,110)\qbezier(-15,110)(5,120)(25,110)
\qbezier(-15,40)(0,60)(25,110)\qbezier(-15,40)(0,90)(25,110)
\qbezier(-15,40)(5,30)(25,40)\qbezier(-15,40)(5,50)(25,40)
\qbezier(25,110)(10,90)(25,40)\qbezier(25,110)(40,90)(25,40)\put(25,40){\line(0,1){70}}
\put(80,75){\circle{1}}\put(60,75){\circle{1}}\put(100,75){\circle{1}}
\qbezier(135,110)(120,90)(135,40)\qbezier(135,110)(150,90)(135,40)\put(135,40){\line(0,1){70}}
\qbezier(135,110)(155,100)(175,110)\qbezier(135,110)(155,120)(175,110)
\qbezier(135,40)(150,60)(175,110)\qbezier(135,40)(150,90)(175,110)
\qbezier(135,40)(155,30)(175,40)\qbezier(135,40)(155,50)(175,40)
\qbezier(175,110)(160,90)(175,40)\qbezier(175,110)(190,90)(175,40)\put(175,40){\line(0,1){70}}
\put(-57,120){s}\put(-37,120){$\infty$}\put(-17,120){$B_1$}\put(3,120){$\infty$}
\put(23,120){$B_2$}\put(133,120){$B_{k-1}$}\put(156,120){$\infty$}\put(173,120){$B_k$}
\put(-57,25){$A_1$}\put(-37,25){$\infty$}\put(-17,25){$A_2$}\put(3,25){$\infty$}
\put(23,25){$A_3$}\put(133,25){$A_{k}$}\put(156,25){$\infty$}\put(173,25){$t$}
\put(-35,115){\vector(1,0){0}}\put(-35,105){\vector(1,0){0}}
\put(-35,45){\vector(1,0){0}}\put(-35,35){\vector(1,0){0}}
\put(5,115){\vector(1,0){0}}\put(5,105){\vector(1,0){0}}
\put(5,45){\vector(1,0){0}}\put(5,35){\vector(1,0){0}}
\put(155,115){\vector(1,0){0}}\put(155,105){\vector(1,0){0}}
\put(155,45){\vector(1,0){0}}\put(155,35){\vector(1,0){0}}
\put(-62,75){\vector(0,-1){0}}\put(-55,75){\vector(0,-1){0}}\put(-48,75){\vector(0,-1){0}}
\put(-22,75){\vector(0,-1){0}}\put(-15,75){\vector(0,-1){0}}\put(-8,75){\vector(0,-1){0}}
\put(128,75){\vector(0,-1){0}}\put(135,75){\vector(0,-1){0}}\put(142,75){\vector(0,-1){0}}
\put(168,75){\vector(0,-1){0}}\put(175,75){\vector(0,-1){0}}\put(182,75){\vector(0,-1){0}}
\put(-32,75){\vector(1,1){0}}\put(-37,80){\vector(1,1){0}}
\put(8,75){\vector(1,1){0}}\put(3,80){\vector(1,1){0}}
\put(158,75){\vector(1,1){0}}\put(153,80){\vector(1,1){0}}
\end{picture}
\end{center}
\caption{$k$-layer non-overlapping zig-zag network: Given the cut
$cut(\{s,B_1,..,B_k\},\{A_1,..,A_k,t\})$, $A_i$ and $B_i$ can
communicate reliably with unbounded rate to $A_{i+1}$ and $B_{i+1}$,
respectively.($s=B_0$, $t=A_{k+1}$). The links from $A_i$ to $B_i$
represent feedback across the cut. This model more accurately
captures the behavior of any cut with $k$ feedback links across the
cut.}\label{fig614}
\end{figure}
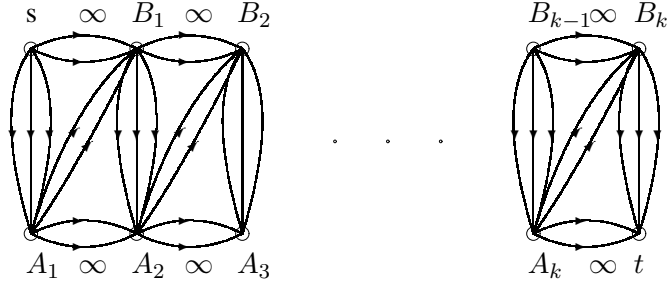

In this section, we consider a special case of the zig-zag network, the
non-overlapping zig-zag network. We present conditions under which
our upper bound is tight and derive a general achievable bound for
non-overlapping zig-zag network.

We call the network shown in Fig.~\ref{fig614} a $k$-layer
non-overlapping zig-zag network. Unlike the general zig-zag network,
feedback transmission is only possible from $A_i$ to $B_i$. $A_i$
and $B_i$ can communicate reliably with unbounded rate to $A_{i+1}$
and $B_{i+1}$, respectivley. ($s=B_0$, $t=A_{k+1}$). Thus, reliable
transmission with unbounded rate is possible from $A_i$ to $A_j$,
and from $B_i$ to $B_j$ for $\forall i<j$. We use $F_i$ and $W_i$ to
denote the set of forward links and feedback links from $B_{i-1}$ to
$A_i$, and from $A_i$ to $B_i$, respectively.  Let $|F_i|=b_i$ and
$|W_i|=m_i$. In this network, we assume that each feedback link from
$A_i$ to $B_i$ has a sufficient capacity to forward all the
information $A_i$ received from $B_{i-1}$. It is clear that the
four-node network is 1-layer non-overlapping zig-zag network. Given
a $k$-layer zig-zag network $\mathcal{G}$, we use $C_{z}$ to denote
the upper bound on $\mathcal{G}$ obtained from Theorem~\ref{bound}.

Now we consider the following strategy for non-overlapping zig-zag
network which is similar to that for a four-node network. We use $C$
to denote the sum of all forward link capacities.

In each round, $s$ and $(B_1,..,B_k)$ together send a $(C,C_z)$
MDS code to $(A_1,..,A_k)$ and $t$ across the cut
$cut(\{s,B_1,..,B_k\},\{A_1,..,A_k,t\})$. For $1\leq i\leq k$, $A_i$
sends its codeword symbols $\hat{W}$ to $B_i$ along each feedback
link using a repetition code. For each feedback link $l$, let $P_l$
denote the information received by $B_i$ on $l$. $B_i$
compares $P_l$ with $W$ which is received from $s$. If $P_l\neq W$,
then $B_i$ obtains a guess $X_l$ identifying the locations of
adversarial links between $B_{i-1}$ and $A_i$ assuming $P_l$ is
reliable. $B_i$ sends claim $(X_l,P_l)$ to $A_{i+1}$ along each link
using repetition code. If $P_l=W$, $B_i$ does not obtain any claim.
For all $2\leq j\leq k$, $A_j$ sends any received claim from
$B_{j-1}$ to the sink reliably. The above strategy is applied in each round. $B_i$ sends claims only when $X_l$ guesses at least
one adversarial forward link which is different from forward links
guessed by $l$ at previous rounds.

For a four-node acyclic network in Fig.~\ref{fig612},
Lemma~\ref{suffbu} shows that our bound is tight when claims are sent
reliably from node $B$ to the sink $t$, i.e., $b\geq 2z+1$. Using
our strategy, we simply extend this result for the non-overlapping
zig-zag network as follows.

\begin{lemma}\label{zig1}
Given a family of $k$-layer non-overlapping zig-zag networks such
that $b_i\geq 2z+1$ for $2\leq i\leq k+1$, rate $C_z$ is achievable.
\end{lemma}

\begin{proof}
Since $b_i\geq 2z+1$ for $2\leq i\leq k+1$, any claim $(X_l,P_l)$
can be sent reliably from $B_{i-1}$ to $A_{i}$ using a repetition
code. Then $A_{i}$ sends this claim reliably to sink $t$. As in the
proof of Lemma~\ref{suffbu}, we first show that at least one
adversarial link is removed whenever sink receives some claim, in
case 1. We also show that correct output is always achievable when
no claims are sent in case 2.

Case 1) sink receives some claim $(X_l,P_l)$.

Assume that feedback link $l$ is between $A_j$ and $B_j$, and $B_j$
sends this claim to $A_{j+1}$. In this case, we use the same
strategy as in the case 1 in Lemma~\ref{suffbu}. Then we show that
the sink removes at least one bad link whenever it receives claim.

Case 2) no claims are sent to the sink.

Similar to the case 2 in the proof of Lemma~\ref{suffbu}, the case
that no claims are sent to the sink occurs only when for each
feedback link $l$ between $A_j$ and $B_j$ either of the following
holds:\begin{itemize}\item the information $B_j$ receives on $l$ is
equal to $W$ where $W$ is the uncorrupted codeword sent by $s$ to
$B_j$\item the guessed set $X_l$  only contains
forward links that have been guessed by $l$ in previous rounds.
From these previous rounds, by case 1, the sink
has already identified as adversarial either $l$ or the guessed
forward links, and is concerned only with the remaining network.
\end{itemize}
Either way,   there
are the following two possibilities for the overall remaining network (recall that $A_j$ transmits $\hat{W}$ to $B_j$).

a) All forward links in ($F_1$,..,$F_k$) and feedback links in
($W_1$,..,$W_k$) are not corrupted.

b) For some $\{i_1,..,i_p\}\subseteq \{1,2,..,k\}$ such that
$m_{i_1}+..+m_{i_p}\leq z$, all feedback links in
($W_{i_1},..,W_{i_p})$ are corrupted and some forward links in
($F_{i_1},..,F_{i_p}$) are corrupted. The furthest downstream
forward links in $F_{k+1}$ can be also corrupted. For $\forall
j\notin\{i_1,..,i_p,k+1\}$, links in $F_j$ and $W_j$ are not
corrupted.

Let $N=\{(i_1,..,i_p)|1\leq i_1<..,<i_p\leq k,
m_{i_1}+..+m_{i_p}\leq z\}\cup \{\emptyset\}$. (Note that
$\{\emptyset\}$ corresponds to the possibility in a)). From a) and
b), there are total $|N|$ possibilities. Exactly only one of them is
true. Now we describe how the correct solution with rate $C_z$ can
be obtained. We check the consistency of the output for each
possibility. For each $(i_1,..,i_p)\in N$, we first remove
$m_{i_1}+..+m_{i_p}$ corresponding feedback links and check whether
there are $K-(z-(m_{i_1}+..+m_{i_p}))$ forward links giving
consistent output such that remaining $(z-(m_{i_1}+..+m_{i_p}))$
forward links are elements of $F_{i_1}\cup ..\cup F_{i_p}\cup
F_{k+1}$. If such a set exists, we denote it by $G(i_1,..,i_p)$. If
there is no such set of $K-(z-(m_{i_1}+..+m_{i_p}))$ forward links
giving consistency, we remove $(i_1,..,i_p)$ from $N$ and ignore
corresponding possibilities.

Now we show that only tuples $(i_1,..,i_p)$ such that
$G(i_1,..,i_p)$ gives the correct output remain in $N$. Since at
least one remaining tuple gives the correct output, it is sufficient
to prove that for any remaining $(i_1,..,i_p)\in N$ and
$(j_1,..,j_r)\in N$, $G(i_1,..,i_p)$ and $G(j_1,..,j_r)$ gives the
same output. This is equivalent to showing that the sum of
capacities of forward links which are contained in both
$G(i_1,..,i_p)$ and $G(j_1,..,j_r)$ is at least $C_z$, i.e.,
\[\sum_{l\in G(i_1,..,i_p)\cap G(j_1,..,j_r)}r(l)\geq C_z.\]
$G(i_1,..,i_p)$ gives $K-(z-(m_{i_1}+..+m_{i_p}))$ forward links
giving consistent output such that remaining
$(z-(m_{i_1}+..+m_{i_p}))$ forward links are in $F_{i_1}\cup ..\cup
F_{i_p}\cup F_{k+1}$. Similarly, $G(j_1,..,j_r)$ gives
$K-(z-(m_{j_1}+..+m_{j_r}))$ forward links giving consistent output
such that remaining $(z-(m_{j_1}+..+m_{j_r}))$ forward links are in
$F_{j_1}\cup ..\cup F_{j_r}\cup F_{k+1}$. In this case, from the
definition of cut-set upper bound in Lemma~\ref{bound2}, the sum of
the capacities of forward links assumed to be correct by both
$G(i_1,..,i_p)$ and $G(j_1,..,j_r)$ is at least $C_z$. Since each
guess gives consistent output, these two guesses gives the same
output. Since any two remaining guesses in $N$ give the same
consistent output, all remaining guesses give the same output.
\end{proof}

\begin{figure*}
\begin{center}
\begin{picture}(200,200)(-40,-20)\centering
\put(-75,40){\circle{5}}\put(-35,40){\circle{5}}\put(5,40){\circle{5}}
\put(115,40){\circle{5}}\put(155,40){\circle{5}}\put(195,40){\circle{5}}\put(235,40){\circle{5}}
\put(-75,110){\circle{5}}\put(-35,110){\circle{5}}\put(5,110){\circle{5}}
\put(115,110){\circle{5}}\put(155,110){\circle{5}}\put(195,110){\circle{5}}\put(235,110){\circle{5}}
\qbezier(-75,110)(-90,90)(-75,40)\qbezier(-75,110)(-60,90)(-75,40)\put(-75,40){\line(0,1){70}}
\qbezier(-75,110)(-55,100)(-35,110)\qbezier(-75,110)(-55,120)(-35,110)
\qbezier(-75,40)(-55,30)(-35,40)\qbezier(-75,40)(-55,50)(-35,40)
\qbezier(-35,110)(-50,90)(-35,40)\qbezier(-35,110)(-20,90)(-35,40)\put(-35,40){\line(0,1){70}}
\qbezier(-35,110)(-15,100)(5,110)\qbezier(-35,110)(-15,120)(5,110)
\qbezier(-35,40)(-15,30)(5,40)\qbezier(-35,40)(-15,50)(5,40)
\qbezier(5,110)(-10,90)(5,40)\qbezier(5,110)(20,90)(5,40)\put(5,40){\line(0,1){70}}
\put(60,75){\circle{1}}\put(40,75){\circle{1}}\put(80,75){\circle{1}}
\qbezier(115,110)(100,90)(115,40)\qbezier(115,110)(130,90)(115,40)\put(115,40){\line(0,1){70}}
\qbezier(115,110)(135,100)(155,110)\qbezier(115,110)(135,120)(155,110)
\qbezier(115,40)(130,60)(155,110)\qbezier(115,40)(130,90)(155,110)
\qbezier(115,40)(135,30)(155,40)\qbezier(115,40)(135,50)(155,40)
\qbezier(155,110)(140,90)(155,40)\qbezier(155,110)(170,90)(155,40)\put(155,40){\line(0,1){70}}
\qbezier(195,110)(180,90)(195,40)\qbezier(195,110)(210,90)(195,40)\put(195,40){\line(0,1){70}}
\qbezier(195,110)(215,100)(235,110)\qbezier(195,110)(215,120)(235,110)
\qbezier(195,40)(210,60)(235,110)\qbezier(195,40)(210,90)(235,110)
\qbezier(195,40)(215,30)(235,40)\qbezier(195,40)(215,50)(235,40)
\qbezier(235,110)(220,90)(235,40)\qbezier(235,110)(250,90)(235,40)\put(235,40){\line(0,1){70}}
\put(-77,120){s}\put(-57,120){$\infty$}\put(-37,120){$B_1$}\put(-17,120){$\infty$}
\put(3,120){$B_2$}\put(113,120){$B_{u-1}$}\put(136,120){$\infty$}\put(153,120){$B_u$}
\put(193,120){$B_{k-1}$}\put(216,120){$\infty$}\put(233,120){$B_k$}
\put(-77,25){$A_1$}\put(-57,25){$\infty$}\put(-37,25){$A_2$}\put(-17,25){$\infty$}
\put(3,25){$A_3$}\put(113,25){$A_{u}$}\put(136,25){$\infty$}\put(153,25){$A_{u+1}$}
\put(193,25){$A_{k}$}\put(216,25){$\infty$}\put(233,25){$A_{k+1}$}
\put(170,75){$\ldots$}
\put(-55,115){\vector(1,0){0}}\put(-55,105){\vector(1,0){0}}
\put(-55,45){\vector(1,0){0}}\put(-55,35){\vector(1,0){0}}
\put(-15,115){\vector(1,0){0}}\put(-15,105){\vector(1,0){0}}
\put(-15,45){\vector(1,0){0}}\put(-15,35){\vector(1,0){0}}
\put(135,115){\vector(1,0){0}}\put(135,105){\vector(1,0){0}}
\put(215,115){\vector(1,0){0}}\put(215,105){\vector(1,0){0}}
\put(215,45){\vector(1,0){0}}\put(215,35){\vector(1,0){0}}
\put(135,45){\vector(1,0){0}}\put(135,35){\vector(1,0){0}}
\put(-82,75){\vector(0,-1){0}}\put(-75,75){\vector(0,-1){0}}\put(-68,75){\vector(0,-1){0}}
\put(-42,75){\vector(0,-1){0}}\put(-35,75){\vector(0,-1){0}}\put(-28,75){\vector(0,-1){0}}
\put(108,75){\vector(0,-1){0}}\put(115,75){\vector(0,-1){0}}\put(122,75){\vector(0,-1){0}}
\put(148,75){\vector(0,-1){0}}\put(155,75){\vector(0,-1){0}}\put(162,75){\vector(0,-1){0}}
\put(188,75){\vector(0,-1){0}}\put(195,75){\vector(0,-1){0}}\put(202,75){\vector(0,-1){0}}
\put(228,75){\vector(0,-1){0}}\put(235,75){\vector(0,-1){0}}\put(242,75){\vector(0,-1){0}}
\put(138,75){\vector(1,1){0}}\put(133,80){\vector(1,1){0}}
\put(218,75){\vector(1,1){0}}\put(213,80){\vector(1,1){0}}
\end{picture}
\end{center}
\caption{Reduced non-overlapping zig-zag network $\mathcal{G}'$ such
that $m_u>z$ and $b_{u+1}\geq 2z+1,\ldots ,b_{k+1}\geq 2z+1$. This
graph is obtained from $\mathcal{G}$ by erasing all feedback links
in $W_1\cup W_2\ldots \cup W_{u-1}$. }\label{fig615}
\end{figure*}
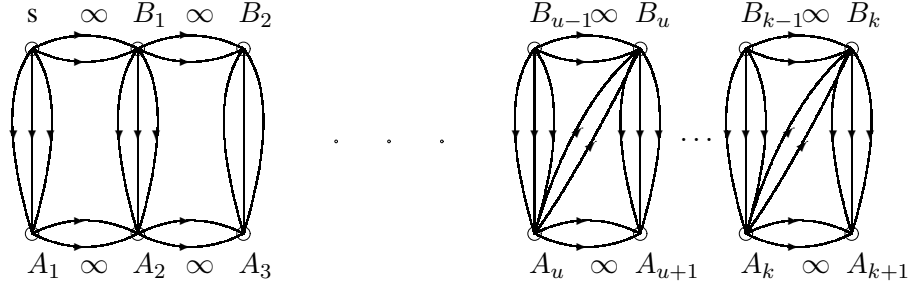

We derive another condition under which our bound is tight.

\begin{lemma}\label{zig2}
Given a family of $k$-layer non-overlapping zig-zag networks such
that $m_u>z$ and $b_{j}\geq 2z+1$ $\forall j\geq u+1$ for any $1\leq
u\leq k$, rate $C_z$ is achievable.
\end{lemma}

\begin{proof}
We consider a reduced non-overlapping zig-zag network $\mathcal{G}'$
shown in Fig.~\ref{fig615} which is obtained from a given $k$-layer
non-overlapping zig-zag network by erasing $m_1+..+m_{u-1}$ feedback
links $W_1\cup ..\cup W_{u-1}$. We use $C_{z}'$ to denote the upper
bound on $\mathcal{G}'$ from Theorem \ref{bound}. Since
$\mathcal{G}'$ is weaker than $\mathcal{G}$, it is sufficient to
show that $C_z'\geq C_z$ and $C_z'$ is achievable on $\mathcal{G}'$.

Step 1) We show that $C_z\leq C_z'$.

We first compute $C_z'$ on $\mathcal{G}'$ from Theorem~\ref{bound}.
Suppose that $C_{z}'$ is obtained by choosing and
$A^*=\{A_{i_1},..,A_{i_p}\}\subseteq \{A_{u+1}..,A_{k}\}$ and
$B^*=\{A_{j_1},..,A_{j_r}\}\subseteq \{A_{u+1},..,A_{k}\}-A^*$ and
applying Lemma~\ref{bound2} after erasing $k$ forward links set
$F^*$, $m$ feedback links set $W^*$. It is sufficient to prove that
choosing the same $F^*$, $W^*$, $A^*$, and $B^*$ on original graph
$\mathcal{G}$ gives the same upper bound $C_z'$.

Since $m_u>z$, $A_u\not\in A^*$ and $A_u\not\in B^*$ from the
definition of upper bound in Lemma~\ref{bound2}. Then
$P_{A^*}\subseteq F_{A_{i_1}}\cup ..\cup F_{A_{i_p}}\subset
F_{u+1}\cup\ldots\cup F_k$, $P_{B^*}\subseteq F_{A_{j_1}}\cup ..\cup
F_{A_{j_r}}\subset F_{u+1}\cup\ldots\cup F_k$.

Since $m_u>z$ and $b_{k+1}>z$, no matter what $W^*$ is erased on
$\mathcal{G}'$, chosen downstream links $R_{A^*}$ and $R_{B^*}$ are
in $F_{u+1}\cup..\cup F_k$, i.e., $R_{A^*},R_{B^*}\in
F_{u+1}\cup\ldots\cup F_k$.

Thus, $Z_{A^*}=P_{A^*}\cup R_{A^*}\subset F_{u+1}\cup..\cup F_k$ and
$Z_{B^*}=P_{B^*}\cup R_{B^*}\subset F_{u+1}\cup..\cup F_k$.

Since all erased forward links in $Z_{A^*}\cup Z_{B^*}$ are in
$F_{u+1}\cup..\cup F_k$ for $\mathcal{G}'$,  erasing the same $F^*$,
$W^*$, $Z_{A^*}$, and $Z_{B^*}$ on original graph $\mathcal{G}$ also
gives the same upper bound $C_z'$ for $\mathcal{G}$. Since $C_z$ is
the minimal upper bound for $\mathcal{G}$, $C_z\leq C_z'$.




Step 2) We show that rate $C_z$ is achievable.

From Lemma~\ref{zig1}, since $b_{u+1}\geq 2z+1,\ldots,b_{k+1}\geq
2z+1$ and $C_z\leq C_z'$, rate $C_z$ is achievable on
$\mathcal{G}'$. Thus, given a non-overlapping zig-zag network
$\mathcal{G}$, we first ignore all feedback links between $A_i$ and
$B_i$ $(1\leq i\leq u-1)$ and apply the same achievable strategy for
$\mathcal{G}'$.

From steps 1) and 2), we complete the proof.
\end{proof}

Now we derive an achievable rate of guess-and-forward strategy for
any non-overlapping zig-zag network.

We use $\mathcal{G}_I$ to denote the non-overlapping zig-zag network
obtained from original $\mathcal{G}$ by erasing all feedback links
in $W_i$ such that $i\not\in I$. Let $b(i,j)=\sum_{u=i+1}^{j}b_i$
denote the number of forward links between $i$ th layer and $j$th
layer. Supersets $P$, $Q$, and $R$ are defined as follows.

\begin{eqnarray*}
P & = & \{\{i\}|1\leq i\leq k\}, \\
Q & = &\{\{i_1,..,i_u\}| \{i_1,..,i_u\}\subset \{1,..,k\},b(1,i_1)\geq 2z+1,\\
  &   &b(i_1,i_2)\geq 2z+1,..,b(i_u,k+1)\geq 2z+1\},\\
R & = &\{\{i_1,..,i_u\}| \{i_1,..,i_u\}\subset \{1,..,k\}, m_{i_1}>z, \\
  &   &b(i_1,i_2)\geq 2z+1,..,b(i_u,k+1)\geq 2z+1\}.\\
\end{eqnarray*}

\begin{lemma}\label{zigzagach}
Given the network in Fig.~\ref{fig614}, rate $\max_{I\in P\cup Q\cup
R} C_I$ is achievable.
\end{lemma}

\begin{proof}
We first show that rate $C_{\{i\}}$ is achievable for $1\leq i\leq
k$. We ignore all feedback links except the feedback links in $W_i$.
Then applying the same achievability strategy for four-node acyclic
network gives the rate $C_{\{i\}}$ from Lemma~\ref{suffbu}
and~\ref{insuffbu}.

For any subset $I\in Q$, we ignore all feedback links except the
feedback links in $W_i$ such that $i\in I$. Then from
Lemma~\ref{zig1}, rate $C_I$ is achievable. Similarly, for any
subset $I\in R$, rate $C_I$ is achievable from Lemma~\ref{zig2}.
This completes the proof.

\end{proof}

\section{Conclusion}
We have studied the capacity of single-source single-sink noiseless
networks under adversarial attack on no more than $z$ edges. In this
work, we have allowed arbitrary link capacities, unlike prior papers. We
have proposed a new cut-set upper bound for the error-correction
capacity for general acyclic networks. This bound tightens previous
cut-set upper bounds. For example networks where the bounds are
tight, we have employed both linear and nonlinear coding strategies
to achieve the capacity. We have 
proved the insufficiency of linear network codes to achieve the
capacity in general. We also have shown by examples that there exist
single-source and single-sink networks for which intermediate nodes
must perform coding, nonlinear error detection or error correction
in order to achieve the network capacity. This is unlike the equal
link capacity case, where coding only at the source suffices to
achieve the capacity of any single-source and single-sink network. We have introduced a
new achievable strategy, guess-and-forward, which is used to show
the capacity of the two-node network and a
family of four-node acyclic networks. Finally, for a class of so called non-overlapping
zig-zag networks, we have derived the rate achieved by
guess-and-forward and presented conditions under which that bound is
tight.

Further work includes characterizing the capacity region of a
four-node acyclic network when the capacity of feedback links is
small. When the lower bound on the feedback link capacity is not
satisfied, we can investigate also the tightness of our bound or find an
achievable capacity region. It would also be interesting to find new achievable strategies and upper bounds for more general zig-zag and other networks, particularly since
cut-set approaches are not sufficient in general~\cite{kosut-polytope}. Investigating networks
for which there exists a gap between known upper and lower bounds may provide further insights. Another related problem, which we treat briefly in our conference paper~\cite{kim2009network}, considers high-probability correction of errors in a causal adversary model as in~\cite{langberg2009binary}.

{\appendix}

\emph{Proof of Lemma~\ref{fdbackmds}} : Since the adversary controls
forward links such that codewords on feedback links are unchanged,
from the definition~\ref{degfree}, the degree of freedom of errors
that the adversary can control for any forward link $l$ is at most
$f(l)$. We prove this lemma by simply extending~\cite[Theorem
4]{cai2006network} which is for the equal link capacities case to
the unequal link capacities case.

Let $M$ denote the transfer matrix whose columns are the coding
vectors assigned to links. Then, the difference set is
\begin{eqnarray*}
\lefteqn{\Delta(V,z)}\\
& = &\{(\theta_l(e)-\theta_l(e'))\cdot M^{-1}:l\in\Gamma_+(t),N(e)\leq z,N(e')\leq z\}\\
& = &\{\theta_l(e-e')\cdot M^{-1}:l\in\Gamma_+(t),N(e)\leq z,N(e')\leq z\}\\
& = &\{\theta_l(d)\cdot M^{-1}:l\in\Gamma_+(t),N(d)\leq 2z\},\\
\end{eqnarray*}
where $N(e)$ denotes the number of links error $e$ occurs and
$\theta_l(e)$ denotes the output of error vector $e$ at the sink
with zero-input.

Last equality comes from $\{e-e':N(e)\leq z,N(e')\leq
z\}=\{d:N(d)\leq 2z\}.$

We use $p$ to denote the maximum number of different error vectors
when the adversary controls $2z$ links. Since $\Delta(V,z) =
\{\theta_l(d)\cdot M^{-1}:l\in\Gamma_+(t),N(d)\leq 2z\}$,
\[|\Delta(V,z)|\leq p\cdot\sum_{i=0}^{2z}{a+b\choose i}.\]

Since $u$ is the sum of $2z$ largest degree of freedom, $p\leq
(q-1)^u$. Thus,
\begin{eqnarray*}
|\Delta(V,z)|&\leq &(q-1)^u\sum_{i=0}^{2z}{a+b\choose i}\le 2^{C}(q-1)^u,\\
\end{eqnarray*}
where $C$ is the sum of all forward link capacities.

After computing the size of the difference set $\Delta(V,z)$, we
apply exactly the same argument as in~\cite[Theorem
4]{cai2006network}. From the argument in the proof of~\cite[Theorem
4]{cai2006network}, it is sufficient to show the existence of
$u\times C$ parity check matrix $H$ such that for all $w\in \Delta^*(V,z)$,
$Hw^{\tau}\neq 0$ where $\Delta^*(V,z)= \Delta(V,z)-\{\textbf{0}\}$.
By using the same technique in~\cite[Theorem 4]{cai2006network}, the
number of matrices $H$ such that there exists $w\in \Delta^*(V,z)$
satisfying $Hw^{\tau}= 0$ is upper bounded by

\begin{eqnarray*}
\lefteqn{q^{u(n-1)}(q-1)^{-1}(\Delta(V,z)-1)}\\
&<&2^{C}q^{u(n-1)}(q-1)^{-1}(q-1)^{u}\\
&= & \frac{2^{C}q^{un}}{q}.\\
\end{eqnarray*}

Thus, if $q$ is a prime power such that $q>2^{C}$, then there must
exist a $u\times C$ matrix such that $Hw^{\tau}\neq 0$ for all $w\in
\Delta^*(V,z)$ and the result follows.

\emph{Proof of Lemma~\ref{suffbu}} :

Since $b\geq 2z+1$, any claim $(X_l,P_l)$ can be sent reliably from
$B$ to $t$ using a repetition code. In case 1, we first show that at
least one adversarial link is removed whenever sink receives some
claim.
In case 2, we show that rate $C_z$ is achievable even when no claims
are sent from $B$.

Case 1) sink receives some claim $(X_i,P_i)$.

The sink compares $P_i$ with $g(\hat{W})$ which is received from $A$
reliably. If $P_i\neq g(\hat{W})$, then feedback link transmitting
$P_i$ is adversarial and the sink ignores it. Otherwise, $P_i$ is
reliable. Since the claim is sent, the sink knows that
$P_i=g(\hat{W})\neq g(W)$ and that guess $X_i$ is correct. Thus the
sink identifies as adversarial the links in $ X_i$,
which are subsequently ignored.
Therefore, in this case, the sink removes at least one bad link whenever $B$
sends claims. 


Case 2) no claims are sent.

From our strategy, the case that no claims are sent from $B$ occurs
only when for each feedback link $l$ either of the following
holds:\begin{itemize}\item the information $B$ receives on $l$ is
equal to $g(W)$ where $W$ is the uncorrupted codeword sent by $s$ to
$B$\item the guessed set $X_l$  only contains
forward links that have been guessed by $l$ in previous rounds.
From these previous rounds, by case 1, the sink
has already identified as adversarial either $l$ or the guessed
forward links, and is concerned only with the remaining network.
\end{itemize}
Either way,   there
are the following three possibilities for the overall remaining network (recall that $A$ transmits $\hat{W}$ to $B$).

(I) All links between $s$ and $A$ and all feedback links are
uncorrupted. Then $\hat{W}=W$ and $g(W)$ is reliably transmitted.

(II) Some links between $s$ and $A$ are corrupted so that $A$
receives $g(\hat{W})\neq g(W)$ from $s$, but the adversary controls
all feedback links such that each feedback link changes $g(\hat{W})$
to $g(W)$.

(III) Some links between $s$ and $A$ are corrupted such that
codewords $A$ sends along each feedback link are unchanged, i.e.,
$\hat{W}\neq W$ and $g(\hat{W})= g(W)$. All feedback links are
reliable and $B$ receives $g(W)$.


If possibility (I) is true, all links between $s$ and $A$ and all
feedback links are uncorrupted. Then there exists a set of $(a+b-z)$
forward links on the cut such that this set includes all $a$ links
between $s$ and $A$ and some $b-z$ links between $B$ and $t$, and
gives consistent output with rate $C_z$. (Note that the sum of
capacities of any $(a+b-z)$ forward links is larger than or equal to
$C_z$ from the definition of our bound in Theorem~\ref{bound}.) If
possibility (II) is true, all $m$ feedback links are corrupted. Then
there exist a set of $(a+b-z+m)$ forward links on the cut that gives
consistent output with rate $C_z$. If possibility (III) is true,
then we obtain the correct output from $(C,C_z)$ MDS code in
Lemma~\ref{fdbackmds}.

Based on the above analysis, we give the following simple decoding
algorithm and prove the correctness of this algorithm as follows.

\begin{algorithm}
\caption{Decoding algorithm for achieving rate $C_z$ when no claims
are sent.} \label{algsuffbu}

\hspace{4mm} \textbf{IF} {there is a set $L_1$ which is composed of
$a+b-z+m$ forward links and gives consistent output,}


\hspace{12mm} \textbf{THEN} the output with rate $C_z$ from $L_1$ is
correct.

\hspace{4mm} \textbf{ELSE IF} {there is a set $L_2$ which is
composed of all $a$ forward links from $s$ to $A$ and some $b-z$
forward links from $B$ to $t$, and gives consistent output,}

\hspace{12mm} \textbf{THEN} the output with rate $C_z$ from $L_2$ is
correct.

\hspace{4mm} \textbf{ELSE} {the output with rate $C_z$ obtained from
$(C,C_z)$ MDS code is correct.}

\hspace{4mm} \textbf{END IF}

\end{algorithm}

Case 2 - a) there is a set $L_1$ composed of $(a+b-z+m)$ forward
links giving consistent output.

In this case, we show that output with rate $C_z$ obtained from
$L_1$ is correct. First we prove that output from $L_1$ is correct
when (II) or (III) is true. If (II) is true, all $m$ feedback links
are corrupted and thus $L_1$ contains at least
$(a+b-z+m)-(z-m)=a+b-2(z-m)$ uncorrupted links. From the definition
of our upper bound in Lemma~\ref{bound1}, the sum of capacities of
any $a+b-2(z-m)$ forward links is larger than or equal to $C_z$.
Since $L_1$ gives consistent output, the output is correct. If (III)
is true, $L_1$ contains at most $z$ corrupted links. In this case,
we can show that $L_1$ gives correct output using
condition~\ref{cond2}. Since (III) is true and the values on
feedback links remain the same, when adversary controls a set of $z$
forward links, $A_1$, the degree of freedom adversary can control is
at most $\sum_{l\in A_1}f(l)$ ($|A_1|=z$). From
condition~\ref{cond2}, the sum of degree of freedom of $z$ links in
$A_1$ plus the sum of capacities of $z-m$ forward links not included
in $L_1$ is less than or equal to $C_z$, i.e., $\sum_{l\in
A_1}f(l)+\sum_{l\in \mathcal{E}-L_1}r(l)\leq C-C_z$. Therefore,
$L_1$ contains at least rate $C-(C-C_z)=C_z$ uncorrupted output and
the output is correct. As shown above, using condition~\ref{cond2},
we have shown that $L_1$ contains at least one correct output if
(III) is true and complete the proof from the consistency of $L_1$.

Case 2 - a - i) there is no set of $(a+b-z)$ links that includes all
$a$ forward links from $s$ to $A$ and gives consistent output.

In this case, possibility (I) cannot hold and thus (II) or (III) is
true. Thus output from $L_1$ is correct.

Case 2 - a - ii) there exists a set $L_2$ composed of $(a+b-z)$
forward links that includes all $a$ links from $s$ to $A$ and gives
consistent output.

We first show that $L_1$ and $L_2$ gives the same consistent output.
$L_1\cap L_2$ is obtained from the cut by erasing $z$ forward links
from $B$ to $t$ that $L_2$ does not include and $z-m$ forward links
$L_1$ does not include. From the definition of our bound in
Lemma~\ref{bound2}, $\sum_{l\in L_1\cap L_2}r(l)\geq C_z$. Thus
$L_1$ and $L_2$ give the same consistent output. Since $L_2$ gives
the correct output when (I) is true, and $L_1$ and $L_2$ give the
same consistent output in this case, output from $L_1$ is correct
when (I) is true. Moreover, we have already shown that $L_1$ gives
the correct output if (II) or (III) is true. Therefore, $L_1$ always
gives correct output.

Case 2 - b) there is no set of $(a+b-z+m)$ forward links giving
consistent output.

In this case, there are more than $z-m$ adversarial forward links on
the cut. Thus (II) cannot hold and (I) or (III) is true. If there is
no set of $(a+b-z)$ forward links that includes all $a$ links from
$s$ to $A$ and gives consistent output, then (III) is true. From
Lemma~\ref{fdbackmds}, output obtained from $(C,C_z)$ MDS code is
correct. Otherwise, suppose that there exists a set $L_2$ composed
of $(a+b-z)$ forward links that includes all $a$ links from $s$ to
$A$ and gives consistent output. We show that output obtained from
$L_2$ is correct.

If possibility (I) is true, since all links between $s$ and $A$ and
all feedback links are uncorrupted, $L_2$ contains at most $z$
corrupted forward links between $B$ and $t$. From the definition of
the Singleton bound, the sum of capacities of $a$ links between $s$
and $A$ plus the sum of any $b-2z$ forward links between $B$ and $t$
is larger than or equal to $C_z$. Thus, from uncorrupted $a$ links
between $s$ and $A$ and some $b-2z$ links between $B$ and $t$ which
are not corrupted, we obtain the correct output rate $C_z$. Since
$L_2$ gives consistent output, the decoded output is correct.

If possibility (III) is true, the adversary controls some forward
links from $s$ to $A$ such that each feedback link transmits $g(W)$,
and $L_2$ contains at most $z$ unknown corrupted links. Using
condition~\ref{cond1}, we will show that the output obtained from
$L_2$ is correct. If $A_2$ is a set of truly corrupted forward
links, the degree of freedom adversary can control is at most
$\sum_{l\in A_2}f(l)$ ($|A_2|\leq z$). From the definition, the
degree of freedom of any forward link from $B$ to $t$ is equal to
the link capacity. From the condition~\ref{cond1}, the sum of degree
of freedom of truly corrupted links in $A_2$ and the sum of $z$
forward links between $B$ and $t$ which are not included in $L_2$ is
less than or equal to $C-C_z$. Therefore, $L_2$ contains at least
$C-(C-C_z)=C_z$ uncorrupted symbols. Since $L_2$ gives consistent
output, the decoded output from $L_2$ is correct. As shown above,
when (III) is true, we can obtain correct output using condition 1.

Therefore, either (I) or (III) is true, $L_2$ gives the correct
output.

\emph{Proof of Lemma~\ref{insuffbu}} :

Since $b\leq 2z$, a claim $(X_l,P_l)$ for any feedback link $l$ is
not reliably transmitted to the sink and adversarial links between
$B$ and $t$ can corrupt this claim arbitrarily. Thus, the sink can
receive different claims on different incoming links. Let $G(l)$ be
the set of distinct claims
$G(l)=\{(X_{l1},P_{l1}),..,(X_{lk},P_{lk}),Y\}$ where $Y$ denotes
that no claims received. Here is the outline of the proof. We first
show that at least one adversarial link is removed except when $b>z$
and the sink receives no claim on all $b$ links for all feedback
links. When $b>z$ and the sink receives no claim on all $b$ links,
since all $b$ links cannot be corrupted at the same time, the sink
knows that $B$ does not send any claim. This case exactly
corresponds to the case 2 in the proof of Lemma~\ref{suffbu} and we
achieve the correct output. This completes the proof. Note that the
same guess-and-forward strategy in
Section~\ref{subsec:fournoderesults} is used.

First we show that any uncorrupted $(a+b-2z)$ forward links between
$s$ and $A$ give the correct decoded output with rate $C_z$. From
the definition of Singleton bound, after erasing $b\leq 2z$ links
between $B$ and $t$ and any set of $2z-b$ links between $s$ and $A$,
the sum of the remaining link capacities are larger than or equal to
$C_z$. Thus any uncorrupted $(a+b-2z)$ links between $s$ and $A$
give the correct message.

Now we assume that $(X_{li},P_{li})$ is received on $n_i$ links and
$Y$ is received on $n_{k+1}$ links $(n_1+..+n_{k+1}=b)$. First we
ignore any $(X_{li},P_{li})$ claiming that there are more than
$z-(b-n_i)$ adversarial links between $s$ and $A$. Since $X_{li}$ is
shown on $n_i$ links, believing $X_{li}$ implies more than $z$
adversarial links on the cut which is a contradiction. Thus, each of
remaining claim $(X_{lj},P_{lj})$ specifies a set $L_j$ which is
composed of at least $(a-(z-(b-n_i)))=a+b-z-n_i$ links between $s$
and $A$ claimed to be correct by $(X_{lj},P_{lj})$. For each such
claim, we check the consistency of the decoded outputs of $L_j$. We
show that if there exist two different claims $(X_{li},P_{li})$ and
$(X_{lj},P_{lj})$ both corresponding to consistent outputs, then
those two outputs should be the same. Since $|L_i|=a+b-z-n_i$,
$|L_j|=a+b-z-n_j$, and $|L_i\cup L_j|\leq a$,
\begin{eqnarray*}
|L_i\cap L_j| & \geq &(a+b-z-n_i)+(a+b-z-n_j)-a\\
&\geq & a+b-2z.\\
\end{eqnarray*}
As we mentioned at the beginning of the proof, the sum of capacities
of any $(a+b-2z)$ link between $s$ and $A$ is larger than or equal
to $C_z$. Therefore $L_i$ and $L_j$ give the same consistent output.

Suppose that we have figured out that a set of links $L$ gives the
correct consistent decoded output. In this case, we add remaining
links not included in $L$ sequentially to $L$, and check the
consistency of any decoded output with rate $C_z$. If outputs are no
more consistent, the added link is adversarial (*).

Now we show that at least one adversarial link is removed except
when $b>z$ and the sink receives no claim on all $b$ links for all
feedback links.

Case 1) all claims are ignored or none of the remaining claims gives
consistent output or all claims $(X_{li},P_{li})$ that give
consistent output satisfy that $P_{li}\neq \hat{W}$

In this case, there are only two possibilities.

(I) feedback link $l$ is adversarial.

(II) feedback link $l$ is reliable and all $b$ links between $B$ and
$t$ are adversarial.

If $b>z$, then possibility (II) cannot hold and feedback link $l$ is
adversarial. We remove it. If $b\leq z$, the sink checks the
consistency of outputs from each set of $(a+b-z)$ links between $s$
and $A$. If no $(a+b-z)$ links set give consistency, then there are
more than $z-b$ adversarial links between $s$ and $A$. Thus
possibility (I) is true and we remove feedback link $l$. Otherwise,
there exists a set $L$ of $(a+b-z)$ links giving consistency. Since
this set contains at most $z$ corrupted links, and thus includes at
least $(a+b-2z)$ uncorrupted links between $s$ and $A$. Then the sum
of capacities of uncorrupted links are larger than or equal to
$C_z$. Thus $L$ gives correct output rate $C_z$. From (*), we can
detect forward adversarial links in this case.

Case 2) there exists a claim $(X_{li},P_{li})$ giving consistent
output and $P_{li}=\hat{W}$.

We show that the output obtained from claim $(X_{li},P_{li})$ should
be correct. If there is at least one uncorrupted link showing
$(X_{li},P_{li})$, then feedback link $l$ is also not corrupted
since $P_{li}=\hat{W}$, and this claim gives correct output rate
$C_z$.  Otherwise, if all $n_i$ links showing this claim are
adversarial, then there are at most $z-n_i$ adversarial
links between $s$ and $A$. Then $L_i$ includes at least
$(a+b-z-n_i)-(z-n_i)=a+b-2z$ uncorrupted links, and thus gives
correct consistent output. From (*), we can also detect adversarial
links in this case.

Case 3) only $Y$ gives consistent output and $b<z$.

In this case, the set of all $a$ forward links from $s$ to $A$ gives
consistent output. The $a$ links between $s$ and $A$ include at least
$a-z\geq a+b-2z$ uncorrupted links since $b<z$. Thus we obtain
correct consistent output from $a$ links and detect adversarial
links from (*).

Case 4) only $Y$ gives consistent output, $b>z$, and at least one of
$b$ links between $B$ and $t$ show claim different from $Y$, i.e.,
$n_{k+1}<b$.

Case 4 - a) $n_{k+1}<b-z$.

If feedback link $l$ is reliable, the links showing claims different
from $Y$ are adversarial. Thus there are more than $b-n_{k+1}>z$
adversarial links and this is a contradiction. Thus feedback link $l$ is
adversarial and we remove it.

Case 4 - b) $b-z\leq n_{k+1}\leq z$.

$Y$ is shown on $n_{k+1}\geq b-z$ links and $b-n_{k+1}\geq b-z$
links show claims different from $Y$. Thus there are at least $b-z$
adversarial links between $B$ and $t$. Then there are at most $2z-b$
adversarial links between $s$ and $A$ and at least $a+b-2z$
uncorrupted links. Thus we also obtain correct output from $a$ links
and use (*) to detect the adversarial links.

Case 4 - c) $z<n_{k+1}<b$.

Since $n_{k+1}>z$, feedback link $l$ transmits $g(\hat{W})=g(W)$ to
$B$ and $B$ does not send any claim. Thus, the links showing claims
different from $Y$ are all adversarial.

For cases 1-4, we have shown that at least one adversarial link is
removed when $b>z$ and the sink receives some claim different from
$Y$ for any feedback link.

To complete the proof it is now sufficient to show that correct
output can be achieved when $b>z$ and the sink receives no claim for
all feedback links $l$. Since $b>z$, at least one link between $B$
and $t$ is uncorrupted. Since all $b$ links show $Y$, this means
that each feedback link transmits $g(\hat{W})=g(W)$ and $B$ does not
send any claim. This case corresponds to the case 2 in Lemma
\ref{suffbu}. Therefore, we can obtain the correct output.



\section*{acknowledgements}

We thank the reviewers for their very insightful and detailed suggestions that were  most helpful in improving this paper.
\bibliographystyle{IEEEtran}

\begin{biographynophoto}{Sukwon Kim} is currently an Senior research engineer at the Digital Media and Communication research center at Samsung Electronics. He received his Ph.D. in 2010 and M.S. degree in 2007 in Electrical Engineering, both from the California Institute of Technology (Caltech). Prior to that, he obtained his B.S. in Electrical Engineering from Seoul National University in 2005.  He has received a Samsung Lee Kun Hee Scholarship from Samsung Scholarship Foundation during his Ph.D study. \end{biographynophoto}
\begin{biographynophoto}{Tracey Ho} (M'06) is an Assistant Professor in
Electrical Engineering and Computer Science at the
California Institute of Technology. She received a
Ph.D. (2004) and B.S. and M.Eng degrees (1999)
in Electrical Engineering and Computer Science
(EECS) from the Massachusetts Institute of Technology
(MIT). 
She was a co-recipient of the  2009 Communications
\& Information Theory Society Joint Paper Award for the paper ``A Random
Linear Network Coding Approach to Multicast." Her primary research interests are in
information theory, network coding and communication
networks. 
\end{biographynophoto}

\begin{biographynophoto}{Michelle Effros} (S'93–-M'95–-SM'03–-F'09) received the B.S. degree with distinction
in 1989, the M.S. degree in 1990, and the Ph.D. degree in 1994, all in
electrical engineering from Stanford University, Stanford, CA.

During the summers of 1988 and 1989, she worked at Hughes Aircraft Company.
She joined the faculty at the California Institute of Technology, Pasadena,
in 1994 and is currently a Professor of Electrical Engineering. Her research interests
include information theory, network coding, data compression, and communications.

Dr. Effros received Stanford's Frederick Emmons Terman Engineering
Scholastic Award (for excellence in engineering) in 1989, the Hughes Masters
Full-Study Fellowship in 1989, the National Science Foundation Graduate
Fellowship in 1990, the AT\&T Ph.D. Scholarship in 1993, the NSF CAREER
Award in 1995, the Charles Lee Powell Foundation Award in 1997, the Richard
Feynman-Hughes Fellowship in 1997, an Okawa Research Grant in 2000,
and was cited by Technology Review as one of the world's top 100 young
innovators in 2002. She and her coauthors received the 2009 Communications
\& Information Theory Society Joint Paper Award for the paper ``A Random
Linear Network Coding Approach to Multicast," which appeared in the IEEE
TRANSACTIONS ON INFORMATION THEORY in October 2006. She is a member
of Tau Beta Pi, Phi Beta Kappa, Sigma Xi, and a fellow of the IEEE. She served
as the Editor of the IEEE Information Theory Society Newsletter from 1995
to 1998 and as a Member of the Board of Governors of the IEEE Information
Theory Society from 1998–-2003 and 2008–-present and has been a member
of the Advisory Committee for the Computer and Information Science and
Engineering (CISE) Directorate at the National Science Foundation from 2009
to the present. She served on the IEEE Signal Processing Society Image and
Multi-Dimensional Signal Processing (IMDSP) Technical Committee from
2001 to 2007 and on ISAT from 2006 to 2009. She was an Associate Editor
for the joint special issue on Networking and Information Theory in the IEEE
TRANSACTIONS ON INFORMATION THEORY and the IEEE/ACM TRANSACTIONS
ON NETWORKING and as an Associate Editor for Source Coding for the IEEE
TRANSACTIONS ON INFORMATION THEORY from 2004 to 2007. She is  an Associate Editor for the special issue of the IEEE
TRANSACTIONS ON INFORMATION THEORY honoring Prof. R. Koetter. She
has served on numerous technical program committees and review boards,
including serving as general cochair for the 2009 Network Coding Workshop.
She will serve as cochair of the Technical Program Committee for the International
Symposium on Information Theory in 2012.
\end{biographynophoto}
\begin{biographynophoto}{Salman Avestimehr} is
an assistant Professor at the School of Electrical and Computer Engineering at Cornell University. He received his Ph.D. in 2008 and M.S. degree in 2005 in Electrical Engineering and Computer Science, both from the University of California, Berkeley. Prior to that, he obtained his B.S. in Electrical Engineering from Sharif University of Technology in 2003. He was also a postdoctoral scholar at the Center for the Mathematics of Information (CMI) at Caltech in 2008. He has received a number of awards including the 2011 Young Faculty Program (YIP) award from the Air Force Office of Scientific Research (AFOSR), the NSF CAREER award (2010), the David J. Sakrison Memorial Prize from the U.C. Berkeley EECS Department (2008), and the Vodafone U.S. Foundation Fellows Initiative Research Merit Award (2005). His research interests include information theory, communications, and networking.
\end{biographynophoto}
 \end{document}